%% file: source.tex
\documentclass[journal]{IEEEtran}
\usepackage{amssymb,stmaryrd,amsmath,amsfonts,rotating}
\usepackage[noadjust]{cite}
\usepackage{color}
\usepackage{epic}
\RequirePackage{bbm}
\input{./definition}
\allowdisplaybreaks
\begin{document}
\title{Polar Codes are Optimal for Lossy Source Coding}
\author{Satish Babu Korada and R{\"u}diger Urbanke\thanks{EPFL, School of Computer,
 \& Communication Sciences, Lausanne, CH-1015, Switzerland,
\{satish.korada, ruediger.urbanke\}@epfl.ch.
This work was partially
supported by the National
Competence Center in Research on Mobile Information and Communication
Systems (NCCR-MICS), a center supported by the Swiss National Science
Foundation under grant number 5005-67322.}
}

\maketitle
\begin{abstract}
We consider lossy source compression of a binary symmetric source
using polar codes and the low-complexity successive encoding
algorithm. It was recently shown by Ar\i kan that polar codes achieve
the capacity of arbitrary symmetric  binary-input discrete memoryless
channels under a successive decoding strategy. We show the equivalent
result for lossy source compression, i.e., we show that this
combination achieves the rate-distortion bound for a binary symmetric source. 
We further show the optimality of polar codes for various problems including the
binary Wyner-Ziv and the binary Gelfand-Pinsker problem.
\end{abstract}

\section{Introduction}
Lossy source compression is one of the fundamental problems of
information theory.  Consider a binary symmetric source (BSS) $Y$.
Let $\disto(\cdot,\cdot)$ denote the Hamming distortion function,
\begin{align*}
\disto(0,0) =\disto(1,1) = 0, \disto(0,1) =1.
\end{align*}
It is well known that in order to compress $Y$ with average distortion
$D$ the rate $R$ has to be at least $R(D) = 1- h_2(D)$, where
$h_2(\cdot)$ is the binary entropy function \cite{Sha59}, \cite[Theorem
10.3.1]{CoT91}.  Shannon's proof of this rate-distortion bound is
based on a random coding argument.  

It was shown by Goblick that in fact linear codes are sufficient
to achieve the rate-distortion bound \cite{Gob62},\cite[Section
6.2.3]{Ber71}.

Trellis based quantizers \cite{ViO74} were perhaps the first ``practical" solution 
to source compression. Their encoding complexity is linear in the blocklength of the code 
(Viterbi algorithm). For any rate strictly
larger than $R(D)$ the gap between the expected distortion and the design distortion $D$ vanishes
exponentially in the constraint length. However, the complexity of the encoding
algorithm also scales exponentially with the constraint length. 

Given the success of sparse graph codes combined with low-complexity message-passing 
algorithms for the channel coding problem, it
is interesting to investigate the performance of such a combination for lossy
source compression.

As a first question, we can ask if the codes themselves are suitable for the task.
In this respect, Matsunaga and Yamamoto \cite{MaY03a} showed that if the degrees
of a low-density parity-check (LDPC) ensemble are chosen as large as $\Theta(\log(N))$, where $N$ is the
blocklength, then  this ensemble
saturates the rate-distortion bound if optimal encoding is employed.
Even more promising, Martininian and Wainwright \cite{WaM07} proved that properly chosen MN codes
with {\em bounded} degrees are sufficient to achieve the rate-distortion bound under optimal encoding.

Much less is known about the performance of sparse graph codes under
{\em message-passing} encoding. In \cite{MaY03} the authors consider
binary erasure quantization, the source-compression  equivalent of
the binary erasure channel (BEC) coding problem.  They show that
LDPC-based quantizers fail if the parity check density is $o(\log(N))$
but that properly constructed low-density generator-matrix (LDGM) based
quantizers combined with message-passing encoders are optimal.
They exploit the close relationship between the channel coding problem
and the lossy source compression problem, together with the fact that LDPC
codes achieve the capacity of the BEC under message-passing decoding, to prove the
latter claim.

Regular LDGM codes were considered in \cite{Mur04}.  Using non-rigorous methods
from statistical physics it was shown that these codes approach rate-distortion
bound for large degrees. It was empirically shown that these codes have good
performance under a variant of belief propagation algorithm (reinforced belief propagation). In \cite{CMZ05} the authors consider check-regular LDGM
codes and show using non-rigorous methods that these codes approach the rate-distortion bound for large check degree.
Moreover, for any rate strictly larger than $R(D)$, the gap between the achieved
distortion and $D$ vanishes exponentially in the check degree. They also
observe that belief propagation inspired decimation (BID) algorithms do not
perform well in this context. 
In \cite{BMZ03}, survey propagation inspired decimation
(SID) was proposed as an iterative algorithm for finding the solutions of 
K-SAT (non-linear constraints) formulae  efficiently. Based on this success, the authors in
\cite{CMZ05} replaced the parity-check nodes with non-linear constraints,
and empirically showed that using SID one can achieve a performance close to the
rate-distortion bound.

The construction in \cite{MaY03} suggests that those LDGM codes whose duals (LDPC)
are optimized for the binary symmetric channel (BSC)
might be good candidates for the lossy compression of a BSS using message-passing encoding.
In \cite{WaM05} the authors consider such LDGM codes and empirically show that by using SID 
one can approach very close to the rate-distortion bound. They also mention that
even BID works well but that it is not as good as SID.  
Recently, in \cite{FiF07} it was experimentally shown that using
BID it {\em is} possible to approach the rate-distortion bound 
closely. The key to making basic BP work well in this context is to choose
the code properly. This suggests that in fact the more sophisticated
algorithms like SID may not even be necessary.

In \cite{GVW08} the authors consider a different approach. They show that for
any fixed $\gamma,\epsilon > 0$
the rate-distortion pair $(R(D)+\gamma,D+\epsilon)$ can be achieved with
complexity $C_1(\gamma) \epsilon^{-C_2(\gamma)}N$. Of course, the complexity
diverges as $\gamma$ and $\epsilon$ are made smaller. The idea there is to
concatenate a small code of rate $R+\gamma$ with expected distortion $D +
\epsilon$. The source sequence is then split into blocks of size equal to the code. 
The concentration with respect to the blocklength implies that under MAP
decoding the probability that the distortion is larger than $D+\epsilon$
vanishes. 

Polar codes, introduced by Ar\i kan in \cite{Ari08}, are the first provably
capacity achieving codes for arbitrary symmetric binary-input discrete
memoryless channels (B-DMC) with low encoding and decoding complexity. 
These codes are naturally suited for decoding via successive cancellation (SC) \cite{Ari08}. 
It was pointed out in \cite{Ari08} that an SC decoder can be
implemented with $\Theta(N\log(N))$ complexity.

We show that polar codes with an SC encoder are also optimal for lossy source compression.
More precisely, we show that 
for any design distortion $0 < D< \frac12$, and any $\delta>0$ and $0 < \beta < \frac12$, 
there exists a sequence of
polar codes of rate at most $R(D)+\delta$ and increasing length $N$ so 
that their expected distortion is at most $D+O(2^{-(N^{\beta})})$.  Their encoding
as well as decoding complexity is $\Theta(N\log(N))$.

\section{Introduction to Polar Codes}
Let $W: \{0,1\} \to \mathcal{Y}$ be a binary-input discrete memoryless channel (B-DMC). 
Let $I(W) \in [0,1]$ denote the mutual information between the input and output of $W$ 
with uniform distribution on the inputs, call it the symmetric mutual
information. 
Clearly, if the channel $W$ is symmetric, then $I(W)$ is the capacity of $W$.
Also, let $Z(W) \in [0,1]$ denote the Bhattacharyya 
parameter of $W$, i.e., $Z(W) = \sum_{y\in\mathcal{Y}} \sqrt{W(y \mid 0)W(y \mid 1)}$.

In the following, an upper case letter $U$ denotes a random variable and 
and $u$ denotes its realization. Let $\bU$ denote the
random vector $(U_0,\dots,U_{N-1})$. For any set $F$,
$|F|$ denotes its cardinality. 
Let $\bU_F$ denote $(U_{i_1},\dots,U_{i_{|F|}})$ and let $\bu_F$ denote
$(u_{i_1},\dots,u_{i_{|F|}})$, where $\{i_k\in F: i_k \leq i_{k+1}\}$. Let
$U_i^j$ denote the random vector $(U_i,\dots,U_j)$ and, similarly, $u_i^j$ denotes
$(u_i,\dots,u_j)$. We use the equivalent notation for other random
variables like $X$ or $Y$. Let Ber$(p)$ denote a Bernoulli random variable with
$\Pr(1) = p$.

The polar code construction is based on the following observation. Let 
\begin{align}\label{eqn:2by2}
G_2=\left[
\begin{array}{cc}
1 & 0 \\
1 & 1
\end{array}
\right].
\end{align}

Let $A_n:\{0,\dotsc,2^n-1\}\to\{0,\dotsc,2^n-1\}$ be a permutation defined
by the bit-reversal operation in \cite{Ari08}. 
Apply the transform $A_n G_2^{\otimes n}$ 
(where ``$\phantom{}^{\otimes n}$'' denotes the $n^{th}$ Kronecker power) to a block of 
$N = 2^n$ bits and transmit the output through independent copies of a B-DMC $W$
(see Figure \ref{fig:transform}). 
 As $n$ grows large, the channels seen by individual bits (suitably defined in
\cite{Ari08}) start \emph{polarizing}: they approach either a noiseless channel
or a pure-noise channel, where the fraction of channels becoming noiseless is
close to the symmetric mutual information $I(W)$.

\begin{figure}[ht]
\centering
\input{ps/transform}
\caption{The transform $A_n G_2^{\otimes n}$ is applied to the information word $\bU$
and the resulting vector $\bX$ is transmitted through the channel $W$. The received
word is $\bY$.}\label{fig:transform}
\end{figure}
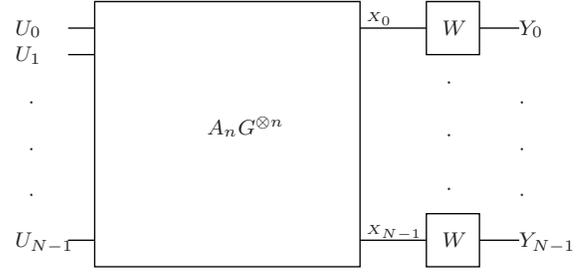

In what follows, let $H_n = A_n G_2^{\otimes n}$. Consider a random vector
$\bU$ that is uniformly distributed over $\{0,1\}^N$. Let $\bX =
\bU H_n$, where the multiplication is performed over GF(2).
Let $\bY$ be the result of sending the components of $\bX$ over the channel $W$.
Let $P(\bU,\bX, \bY)$ denote the induced probability distribution on the set
$\{0,1\}^N\times \{0,1\}^N \times \mathcal{Y}^N$. The
channel between $\bU$ and $\bY$ is defined by the transition
probabilities  
\begin{align*}
P_{\bY\mid \bU}(\by\mid \bu) = \prod_{i=0}^{N-1} W(y_i\mid x_i) =
\prod_{i=0}^{N-1} W(y_i\mid (\bu H_n)_i).
\end{align*}

Define $W^{(i)}: \{0,1\} \to \mathcal{Y}^N \times \{0,1\}^{i-1}$ as the
channel with input $u_i$, output $(y_0^{N-1},u_0^{i-1})$, and transition
probabilities given by 
\begin{align}\label{eqn:transitionWi}
W^{(i)}(\by,u_0^{i-1}\mid u_i) & \triangleq P(\by,u_0^{i-1}\mid u_i) \nonumber \\
& =  \sum_{u_{i+1}^{N-1}}\frac{P(\by \mid \bu) P(\bu)}{P(u_i)}\nonumber\\
&=  \frac{1}{2^{N-1}} \sum_{u_{i+1}^{N-1}}P_{\bY\mid\bU}(\by\mid \bu).
\end{align}
Let $Z^{(i)}$ denote the Bhattacharyya parameter of the channel $W^{(i)}$,
\begin{align}\label{eqn:channelZ}
Z^{(i)} & = \sum_{y_0^{N-1},u_0^{i-1}} \sqrt{W^{(i)}(y_0^{N-1},u_0^{i-1}\mid
0)W^{(i)}(y_0^{N-1},u_0^{i-1}\mid 1)}.
\end{align}

The SC decoder operates as follows:
the bits $U_i$ are decoded in the order $0$ to $N-1$. 
The likelihood of $U_i$ is computed using the channel law
$W^{(i)}(\by,\hu_0^{i-1}\mid u_i)$, where $\hat{u}_0^{i-1}$ are the estimates of
the bits $U_0^{i-1}$ from the previous decoding steps.  

In \cite{Ari08} it was shown that the fraction of the channels $W^{(i)}$ that
are approximately noiseless approaches $I(W)$. More precisely, it was shown that
the $\{Z^{(i)}\}$ satisfy
\begin{align}\label{eqn:Zestimate1}
\lim_{n\to \infty}\frac{\lvert\left\{i\in\{0,\dotsc,2^n-1\}: Z^{(i)} < 
2^{-{\frac{5n}{4}}} \right\} \rvert}{2^n} = I(W).
\end{align}

In \cite{ArT08}, the above result was significantly strengthened to
\begin{align}\label{eqn:Zestimate2}
\lim_{n\to \infty}\frac{\lvert\left\{i\in\{0,\dotsc,2^n-1\}: Z^{(i)} < 
2^{-2^{n\beta}} \right\} \rvert}{2^n} = I(W),
\end{align}
which is valid for any $0 \leq \beta < \frac12$.

This suggests to use these noiseless channels (i.e., those channels at position
$i$ so that $Z^{(i)} <2^{-2^{n\beta}}$) for transmitting information while
fixing the symbols transmitted through the remaining channels to a value known both to
sender as well to the receiver. 
Following Ar\i kan, call those components $U_i$ of $\bar{U}$ which are fixed 
``frozen," (denote this set of positions as $F$) and the remaining ones ``information" bits. 
If the channel $W$ is symmetric we can assume without loss of
generality that the fixed positions are set to $0$. 
In \cite{Ari08} it was shown that the block error probability of the SC decoder
is bounded by $\sum_{i\in F}Z^{(i)}$, which is of order $O(2^{-2^{n\beta}})$
for our choice. 
Since the fraction of approximately noiseless channels tends to $I(W)$, this scheme
achieves the capacity of the underlying symmetric B-DMC $W$.

In \cite{Ari08} the following alternative interpretation was mentioned;
the above procedure can be seen as transmitting a codeword 
of a code defined through its generator matrix as follows. 
A polar code of dimension $0\leq k \leq 2^n$ is defined by choosing a subset of
the rows of $H_n$ as the generator matrix.
The choice of the generator vectors is based on the values of $Z^{(i)}$.
A polar code is then defined as the set of codewords of the form $\bar{x} =
\bar{u}H_n$, where the bits $i\in F$ are fixed to $0$.
The well known Reed-Muller codes can be considered as special cases of polar codes with
a particular rule for the choice of $F$.

Polar codes with SC decoding have an interesting, and of as yet not fully
explored, connection  to the recursive decoding of Reed-Muller codes as
proposed by Dumer \cite{Dum04}. The Plotkin $(u, u+v)$ construction in
Dumer's algorithm plays the role of the channel combining and channel
splitting for polar codes. Perhaps the two most important differences
are  (i) the construction of the code itself (how the frozen vectors
are chosen), and (ii) the actual decoding algorithm and the order in
which information bits are decoded. A better understanding of this connection
might lead to improved decoding algorithms for both constructions.

\begin{figure}[htp]
\centering
\input{ps/trellis_rev}
\caption{\label{fig:trellis} Factor graph representation used by the SC decoder.
$W(y_i\mid x_i)$ is the initial prior of the variable $X_i$, when $y_i$ is received at the 
output of a symmetric B-DMC $W$.}
\end{figure}

To summarize, the SC decoder operates as follows.

For each $i$ in the range $0$ till $N-1$:
\begin{itemize}
\item[(i)] If $i\in F$, then set $u_i = 0$.
\item[(ii)] If $i\in F^c$, then compute \begin{align*} l_i(\by,u_0^{i-1})
= \frac{W^{(i)}(\by, u_0^{i-1}\mid u_i=0)}{W^{(i)}(\by, u_0^{i-1}\mid
u_i=1)} \end{align*} and set \begin{align}\label{eqn:mapdecision}
u_i = \left\{ \begin{array}{cc} 0, & \text{ if } l_i > 1,\\ 1, &
\text{ if } l_i \leq1.  \end{array}\right.  \end{align} \end{itemize}

As explained in \cite{Ari08} using the factor graph representation shown in
Figure~\ref{fig:trellis}, the SC decoder can be implemented with complexity
$\Theta(N\log(N))$. A similar representation was considered for decoding of
Reed-Muller codes by Forney in \cite{For01}.

\subsection{Decimation and Random Rounding}
In the setting of channel coding there is typically one codeword (namely
the transmitted one) which has a posterior that is significantly larger than all other codewords. 
This makes it possible for a greedy message-passing algorithm to 
successfully move towards this codeword in small steps, using 
at any given moment ``local" information provided by the decoder.

In the case of lossy source compression there are typically many
codewords that, if chosen, result in similar distortion. Let us assume that 
these ``candidates'' are roughly uniformly spread
around the source word to be compressed. It is then clear that a local
decoder can easily get ``confused," producing locally conflicting
information with regards to the ``direction" into which one should
compress.

A standard way to overcome this problem is to combine the message-passing
algorithm with {\em decimation} steps. This works as follows; first run the
iterative algorithm for a fixed number of iterations and subsequently
decimate a small fraction of the bits. More precisely, this means
that for each bit which we decide to decimate we choose a {\em value}.
We then remove the decimated variable nodes and adjacent edges from
the graph. One is hence left with a {\em smaller} instance of
essentially the same problem. The same procedure is then repeated
on the reduced graph and this cycle is continued until all variables
have been decimated.

One can interpret the SC operation as a kind of decimation where
the order of the decimation is fixed in advance ($0,\dots,N-1$).
In fact, the SC decoder can be interpreted as a particular instance
of a BID.

When making the decision on bit $U_i$ using the SC decoder, it is
natural to choose that value for $U_i$ which maximizes the posterior.
Indeed, such a scheme works well in practice for source compression.
For the analysis however it is more convenient to use {\em randomized
rounding}.  In each step, instead of making the MAP decision we
replace \eqref{eqn:mapdecision} with
\begin{align*}
u_i = \left\{ \begin{array}{cc}
0, & \text{ w.p. } \frac{l_i}{1+l_i},\\
1, & \text{ w.p. } \frac{1}{1+l_i}.
\end{array}\right.
\end{align*} 
In words, we make the decision proportional to the likelihoods.
Randomized rounding as a decimation rule is not new. E.g., in
\cite{MRS07} it was used to analyze the performance of BID for
random $K$-SAT problems.

For lossy source compression, the SC operation is
employed at the encoder side to map the source vector to a codeword.
Therefore, from now onwards we refer to this operation as SC {\em encoding}.

\section{Main Result}
\subsection{Statement}
\begin{theorem}[Polar Codes Achieve the Rate-Distortion Bound]\label{the:main}
Let $Y$ be a BSS and fix the {\em design} distortion $D$, $0 < D < \frac12$.
For any rate $R > 1-h_2(D)$ and any $0 < \beta < \frac12$,
there exists a sequence of polar codes of length $N$ with rates $R_N <R$  
so that under SC encoding using randomized rounding they
achieve expected distortion $D_N$ satisfying
\begin{align*}
D_N & \leq D + O(2^{-(N^\beta)}). 
\end{align*}
The encoding as well as decoding complexity of these codes is $\Theta(N
\log(N))$.  \end{theorem} 
\subsection{Simulation Results and Discussion} 
Let us consider how polar codes behave in practice.
Recall that the length $N$ of the code is always a power of $2$, i.e., $N = 2^n$.
Let us construct a polar code to achieve a distortion $D$. 
Let $W$ denote the channel BSC$(D)$ and let $R=R(D) + \delta$ for some $\delta >0$.

In order to fully specify the code we need to specify the set $F$, i.e., the set
of frozen components. We proceed as follows.
First we estimate the $Z^{(i)}$s for all $i\in\{0,N-1\}$ and sort
the indices $i$ in decreasing order of $Z^{(i)}$s. The set $F$
consists of the first $RN$ indices, i.e., 
it consists of the indices corresponding to the $RN$ largest $Z^{(i)}$s. 

This is similar to the channel code construction for the BSC$(D)$ but there
is a slight difference. For the case of channel coding we assign all indices $i$
so that $Z^{(i)}$ is {\em very} small, i.e., so that lets say $Z^{(i)}<\delta$, 
to the set $F^c$. Therefore,
the set $F$ consists of all those indices $i$ so that $Z^{(i)}\geq \delta$.

For the source compression, on the other hand, $F$ consists of all those
indices $i$ so that  $Z^{(i)}\geq 1-\delta$, i.e., of all those indices
corresponding to {\em very} large values of $Z^{(i)}$.

Putting it differently, in channel coding, the rate $R$ is chosen to
be strictly less than $1-h_2(D)$, whereas in source compression it
is chosen so that it is strictly larger than this quantity.
Figure~\ref{fig:LOSSY} shows the performance of the SC encoding algorithm combined with
randomized rounding. As asserted by Theorem~\ref{the:main},  the points approach 
the rate-distortion bound as the block length increases.

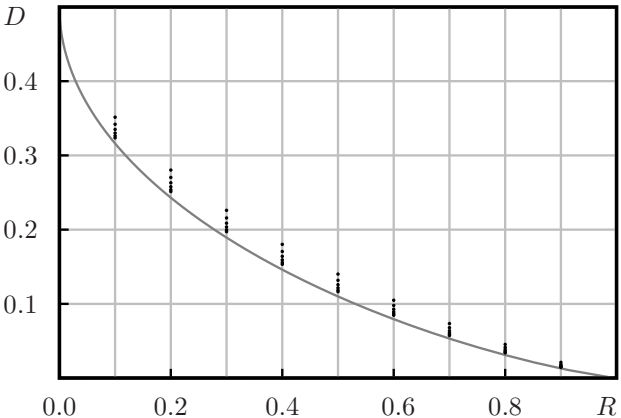
\begin{figure}[htp]
\centering
\input{ps/LOSSY}
\caption{\label{fig:LOSSY} The rate-distortion performance for the SC encoding
algorithm with randomized rounding for $n=9,11,13,15,17$ and $19$. As the block length increases the points move closer to the rate-distortion bound.}
\end{figure}

In \cite{HKU09} the performance of polar codes for lossy source compression was
already investigated empirically. Note that the construction used in \cite{HKU09} is 
different from the current construction. Let us recall. Consider a BSC$(p)$, where
$p = h_2^{-1}(1-h_2(D))$. Let the corresponding Bhattacharyya constants be 
$\tilde{Z}^{(i)}$s. In \cite{HKU09} first a channel code 
of rate $1-h_2(p) - \epsilon$ is constructed according to the values $\tilde{Z}^{(i)}$s.
Let $\tilde{F}$ be the corresponding frozen set. The set $F$ for the source code is given by 
\begin{align*}
F = \{N-1-i: i \in \tilde{F}^c\}.
\end{align*}
The rationale behind this construction is that the resulting source
code is the dual of the channel code designed for the BSC$(p)$. The
rate of the resulting source code is equal to $h_2(p) + \epsilon  =
1-h_2(D) + \epsilon$.  Although this code construction is different,
empirically the resulting frozen sets are very similar.

There is also a slight difference with respect to the  decimation
algorithm.  In \cite{HKU09} the decimation step is based on MAP
estimates, whereas in the current setting we use randomized rounding.

Despite all these differences the performance of both schemes is
comparable.

\section{The Proof}

From now on we restrict $W$ to be a BSC$(D)$, i.e.,
\begin{align*}
W(0\mid 1) &= W(1\mid 0) = D, \\
W(0\mid 0) &= W(1\mid 1) = 1-D.
\end{align*}
As immediate consequence we have 
\begin{align}\label{equ:channelsymmetry}
W(y\mid x) = W(y \oplus z \mid x \oplus z).
\end{align}
This extends in a natural way if we consider vectors. 

\subsection{The Standard Source Coding Model}
Let us describe lossy source compression using polar codes in more detail.
We refer to this as the ``Standard Model.''
In the following we assume that we want to compress the source with
average distortion $D$.

{\em Model: }
Let $\by=(y_0,\dots,y_{N-1})$ denote $N$ i.i.d. realizations of
the source $Y$. Let $F\subseteq \{0,\dots, N-1\}$ and let $\tilde{u}_F
\in\{0,1\}^{|F|}$ be a fixed vector. In the sequel we use the shorthand
``SM$(F,\tilde{u}_F)$" to denote the Standard Model with frozen set $F$ whose
components are fixed to $\tilde{u}_F$.  It is defined as follows.

{\em Encoding: }
Let $f^{\tilde{u}_F}:\{0,1\}^N \to
\{0,1\}^{N-|F|}$ denote the {\em encoding} function. 
For a given $\by$ we first compute $\bu$,
as described below, where $\bu = (u_0, \dots, u_{N-1})$.
Then $f^{\tilde{u}_F}(\by)=\bu_{F^c}$.
 
Given $\by$, for each $i$ in the range $0$ till $N-1$:
\begin{itemize}
\item[(i)] Compute
\begin{align*}
\l_i(\by,u_0^{i-1}) \triangleq\frac{W^{(i)}(\by, u_0^{i-1}\mid
u_i=0)}{W^{(i)}(\by, u_0^{i-1}\mid u_i=1)}.
\end{align*}
\item[(ii)] If $i \in F^c$ then set $u_i=0$ with probability
$\frac{l_i}{1+l_i}$
and equal to $1$ otherwise; if  $i \in F$ then set $u_i = \tilde{u}_i$.
\end{itemize}

{\em Decoding: }
The decoding function $\hat{f}^{\tilde{u}_F}: \{0,1\}^{N-|F|} \to
\{0,1\}^{N}$ maps $\bu_{F^c}$ back to the 
{\em reconstruction point} $\bx$ via $\bx = \bu H_n$, where $\bu_{F}=\tilde{u}_F$. 

{\em Distortion: }
The average distortion incurred by this scheme is given by $\bE[\disto(\bY,\bX)]$,
where the expectation is over the source randomness and the randomness involved
in the randomized rounding at the encoder.

{\em Complexity:} The encoding (decoding) task for source coding is the same as the decoding (encoding) task for
channel coding. As remarked before, both have complexity $\Theta(N\log N)$. 

Remark: Recall that $l_i$ is the posterior of the variable $U_i$ 
given the observations $\bY$ as well as $\bU_0^{i-1}$, under
the assumption that $\bU$ has uniform prior and that $\bY$ is the result of
transmitting $\bU H_n$ over a BSC$(D)$.

\subsection{Computation of Average Distortion}
The encoding function $f^{\tilde{u}_F}$ is random.
More precisely, in step $i$ of the encoding process, $i \in F^c$, 
we fix the value of $U_i$ proportional to the posterior (randomized rounding) $P_{U_i\mid
U_0^{i-1},\bY}(u_i\mid u_0^{i-1},\by)$. 
This implies that the probability of picking a vector $\bu$ given $\by$ is equal to
\begin{align*}
\begin{cases}
0, & \bu_F \neq \tilde{u}_F, \\
\prod_{i\in F^c} P_{U_i\mid
U_0^{i-1},\bY}(u_i\mid u_0^{i-1},\by), & \bu_F = \tilde{u}_F.
\end{cases}
\end{align*}
Therefore, the average (over $\by$ and the randomness of the encoder) 
distortion of SM$(F,\tilde{u}_F)$ is given by 
\begin{align}\label{eqn:distortionF}
D_N(F,\tilde{u}_F) = &\sum_{\by\in\{0,1\}^N} \frac{1}{2^N} \sum_{\bu_{F^c} \in \{0,1\}^{|F^c|}} 
\nonumber\\ 
&\prod_{i\in F^c} P(u_i\mid u_0^{i-1},\by)\disto(\by,\bu H_n),
\end{align}
where $U_i = \tilde{u}_i$ for $i\in F$.

We want to to show that there exists a set $F$ of cardinality roughly $N
h_2(D)$ and a vector $\tilde{u}_F$ such that
 $D_N(F,\tilde{u}_F) \approx
D$. This will show that polar codes
 achieve the rate-distortion bound.

For the proof it is more convenient not to determine the distortion
for a fixed choice of $\tilde{u}_F$ but to compute the average distortion 
over all possible choices (with a uniform distribution over
these choices). Later, in Section~\ref{sec:gauge}, we will see that the distortion {\em does not depend} on
the choice of $\tilde{u}_F$. A convenient choice is therefore to set it to zero. 
This will lead to the desired final result.

Let us therefore start by computing the {\em average distortion}.
Let $D_N(F)$ denote the distortion obtained by averaging $D_N(F,\tilde{u}_F)$ over all 
$2^{|F|}$ possible values of $\tilde{u}_F$. We will show that
$D_N(F)$ is close to $D$.

The distortion $D_N(F)$ can be written as
\begin{align*}
D_N(F) = &
\sum_{\tilde{u}_F \in \{0,1\}^{|F|}} 
\frac{1}{2^{|F|}} 
D_N(F,\tilde{u}_F)\\
= &\sum_{\tilde{u}_F}
\frac{1}{2^{|F|}} 
\sum_{\by} \frac{1}{2^N}
\\
&\sum_{\bu_{F^c}} \prod_{i\in F^c} P(u_i\mid u_0^{i-1},\by)\disto(\by,\bu H_n)\\
= &\sum_{\by} \frac{1}{2^N} \sum_{\bu} \frac{1}{2^{|F|}} \prod_{i\in F^c} P(u_i\mid u_0^{i-1},\by)\disto(\by,\bu H_n).
\end{align*}
Let $Q_{\bU, \bY}$ denote the distribution defined by $Q_{\bY}(\by) =
\frac{1}{2^N}$ and $Q_{\bU \mid \bY}$ defined by
\begin{align}\label{eqn:distQ}
Q(u_i\mid u_0^{i-1},\by) = \left\{
\begin{array}{cc}
\frac12, & \text{ if }i\in F, \\
P_{U_i\mid U_0^{i-1},\bY}(u_i\mid u_0^{i-1},\by), & \text{ if } i\in F^c.
\end{array}\right.
\end{align}
Then, 
\begin{align*}
D_N(F) = \bE_{Q}[{\disto}(\bY,\bU H_n)],
\end{align*}
where $\bE_Q[\cdot]$ denotes expectation with respect to the distribution $Q_{\bU, \bY}$.

Similarly, let $\bE_P[ \cdot ]$ denote the expectation with respect to the distribution
$P_{\bU,\bY}$. Recall that $P_{\bY}(\by) = \frac{1}{2^N}$ and that
we can write  $P_{\bU\mid \bY}$ in the form
\begin{align*}
P_{\bU\mid \bY}(\bu \mid \by)  = \prod_{i=0}^{N-1} 
P_{U_i\mid U_0^{i-1},\bY}(u_i\mid u_0^{i-1},\by).
\end{align*}
If we compare $Q$ to $P$ we see that they have the same 
structure except for the components $i \in F$.
Indeed, in the following lemma we show that the total variation distance between $Q$
and $P$ can be bounded in terms of how much the posteriors $Q_{U_i\mid U_0^{i-1},\bY}$ 
and $P_{U_i\mid U_0^{i-1},\bY}$ differ for $i \in F$.

\begin{lemma}[Bound on the Total Variation Distance]\label{lem:totalvarbnd}
Let $F$ denote the set of frozen indices and let the probability distributions
$Q$ and $P$ be as defined above. Then
\begin{align*}
&\sum_{\bu,\by}|Q(\bu,\by)-P(\bu,\by)| \\
& \phantom{xxxxx} \leq 2 \sum_{i\in F}\bE_P\left[\Big\vert \frac12 - P_{U_i\mid U_0^{i-1},\bY}(0\mid
U_0^{i-1},\bY)\Big\vert\right].
\end{align*}
\end{lemma}
\begin{proof}
\begin{align*}
&\sum_{\bu}\vert Q(\bu\mid \by) - P(\bu\mid\by)\vert \\
&=\sum_{\bu}\Big\vert\prod_{i=0}^{N-1}  Q(u_i\mid u_0^{i-1},\by) -
\prod_{i=0}^{N-1}P(u_i\mid u_0^{i-1},\by)\Big\vert\\
 & = \sum_{\bu}\Big\vert \sum_{i=0}^{N-1} \Bigl[ \bigl( Q(u_i\mid u_0^{i-1},\by)-P(u_i\mid
u_0^{i-1},\by)\bigr) \; \cdot \\
& \phantom{==} \Bigl(\prod_{j=0}^{i-1} P(u_j\mid u_0^{j-1},\by) \Bigr) 
\Bigl(\prod_{j=i+1}^{N-1} Q(u_j\mid u_0^{j-1},\by)\Bigr) \Bigr]\Big\vert.
\end{align*}
In the last step we have used the following telescoping expansion:
\begin{align*}
A_0^{N-1} - B_0^{N-1} & = \sum_{i=0}^{N-1} A_{0}^{i} B_{i+1}^{N-1} - \sum_{i=0}^{N-1} A_{0}^{i-1} B_{i}^{N-1},
\end{align*}
where $A_{k}^{j}$ denotes here the product $\prod_{i=k}^{j} A_i$.

Now note that if $i \in F^c$ then $ Q(u_i\mid u_0^{i-1},\by)=P(u_i\mid
u_0^{i-1},\by)$, so that these terms vanish.
The above sum therefore reduces to
\begin{align*}
& \phantom{=} \sum_{\bu}\Big\vert\sum_{i \in F} \Bigl[ 
\underbrace{\bigl( Q(u_i\mid u_0^{i-1},\by)-P(u_i\mid
u_0^{i-1},\by)\bigr)}_{\leq \mid \frac12 -P(u_i\mid u_0^{i-1},\by)\mid }
\; \cdot \\
& \phantom{==} \Bigl(\prod_{j=0}^{i-1} P(u_j\mid u_0^{j-1},\by) \Bigr) 
\Bigl(\prod_{j=i+1}^{N-1} Q(u_j\mid u_0^{j-1},\by)\Bigr) \Bigr]\Big\vert \\
& \leq \sum_{i \in F} \sum_{\bu_0^{i}}
\Big\vert \frac12 -P(u_i\mid u_0^{i-1},\by)\Big\vert
\prod_{j=0}^{i-1} P(u_j\mid u_0^{j-1},\by) \\
& \leq  2 \sum_{i\in F} \bE_{P_{\bU\mid\bY = \by}}\left[\Big\vert \frac12 - P_{U_i\mid U_0^{i-1},\bY}(0\mid
U_0^{i-1},\by)\Big\vert\right].
\end{align*}
In the last step the summation over $u_i$ gives rise to the factor $2$, whereas
the summation over $u_{0}^{i-1}$ gives rise to the expectation.

Note that $Q_{\bY}(\by) = P_{\bY}(\by) = \frac{1}{2^N}$. The claim follows by taking
the expectation over $\bY$.
\end{proof}

\begin{lemma}[Distortion under $Q$ versus Distortion under $P$]\label{lem:QversusP}
Let $F$ be chosen such that for $i \in F$ 
\begin{align}\label{equ:QversusP}
\bE_P\left[\Big|\frac12 - P_{U_i\mid U_0^{i-1},\bY}(0\mid
U_0^{i-1},\bY)\Big|\right] \leq \delta_N.
\end{align}
The average distortion is then bounded by 
\begin{align*}
\frac{1}{N}\bE_{Q}[\disto(\bY,\bU H_n)] \leq \frac{1}{N}\bE_{P}[\disto(\bY,\bU
H_n)] + \vert F\vert 2\delta_N. 
\end{align*}
\end{lemma}
\begin{proof}
\begin{align*}
\bE_Q&[\disto(\bY, \bU H_n)] - \bE_P[\disto(\bY, \bU H_n)] \\
& =  \sum_{\bu,\by} \Big(Q(\bu,\by) - P(\bu,\by) \Big)\disto(\by, \bu H_n)\\
& \leq  N \sum_{\bu,\by} \Big\vert Q(\bu,\by) -
P(\bu,\by) \Big\vert\\
& \stackrel{\text{Lem.}~\ref{lem:totalvarbnd}}{\leq} 2 N\sum_{i\in F}
\bE_P\left[\Big\vert \frac12 - P_{U_i\mid U_0^{i-1},\bY}(0\mid U_0^{i-1},\bY)\Big\vert\right]\\
& \leq |F| 2 N \delta_N .
\end{align*}
\end{proof}

From Lemma~\ref{lem:QversusP} we see that the average (over $\by$ as well as $\tilde{u}_F$) distortion of the
Standard Model is upper bounded by the average distortion with respect
to $P$ plus a term which bounds the ``distance" between $Q$ and
$P$.

\begin{lemma}[Distortion under $P$]\label{lem:distortionunderP}
\begin{align*}
\bE_{P}[\disto(\bY,\bU H_n)] = N D.
\end{align*}
\end{lemma}
\begin{proof}
Let $\bX=\bU H_n$ and write
\begin{align*}
\bE_{P}&[\disto(\bY, \bU H_n)]  \\
& = \sum_{\bu, \by} P_{\bU,\bY} (\bu,\by) \;  \disto(\by,\bu H_n) \\
& =\sum_{\by, \bu, \bx} P_{\bU, \bX , \bY} (\bu, \bx , \by)
\; \disto(\by, \bu H_n) \\
& = \sum_{\by, \bu, \bx} P_{\bX ,\bY} (\bx,\by)  
\underbrace{P_{\bU \mid \bX, \bY} (\bu \mid \bx, \by)}_{\text{$\{0,
1\}$-valued}} \; \disto(\by,\bx) \\
& = \sum_{\by,\bx} P_{\bX , \bY} (\bx, \by) \; \disto(\by,\bx).
\end{align*}
Note that the unconditional distribution of $\bX$ as well as $\bY$ is the uniform one and that the 
channel between $\bX$ and $\bY$ is memoryless and identical for each component. Therefore, we can write 
this expectation as
\begin{align*}
\bE_{P}[\disto(\bY,\bU H_n)] 
& =N \sum_{x_0,y_0} P_{X_0,Y_0}(x_0 ,y_0)  \; \disto(y_0,x_0) \\ 
& \stackrel{(a)}{=} N \sum_{x_0 }P_{X_0}(x_0) \sum_{y_0} W(y_0 \mid x_0 ) \;
\disto(y_0,x_0) \\
& = N W (0 \mid 1 ) \stackrel{(b)}{=} N D.
\end{align*}
In the above equation, $(a)$ follows from the fact that $P_{Y\mid X}(y\mid x) = W(y\mid
x)$, and $(b)$ follows from our assumption that $W$ is a BSC$(D)$.
\end{proof}

This implies that if we use all the variables $\{U_i\}$ to
represent the source word, i.e., $F$ is empty, then the algorithm results in an average distortion
$D$. But the rate of such a code would be $1$. Fortunately, the last problem is
easily fixed. If we choose $F$ to consist of
those variables which are ``essentially random,'' then there is only a small distortion penalty
(namely, $|F|2 \delta_N$) to pay with respect to the previous case.  
But the rate has been decreased to $1-|F|/N$.

Lemma~\ref{lem:QversusP} shows that the guiding principle for choosing 
the set $F$ is to include the indices with small $\delta_N$ in
\eqref{equ:QversusP}. In the following lemma, we find a sufficient condition for
an index to satisfy \eqref{equ:QversusP}, which is easier to handle. 
\begin{lemma}[$Z^{(i)}$ Close to $1$ is Good]\label{lem:Zdelta}
If $Z^{(i)}  \geq 1- 2\delta_N^2$,
then 
\begin{align*}
\bE_P\left[\Big| \frac12 - P_{U_i\mid U_0^{i-1},\bY}(0\mid U_0^{i-1},\bY)\Big|\right] \leq
\delta_N.
\end{align*}  
\end{lemma}
\begin{proof}
\begin{align*}
&\bE_P\left[\sqrt{P_{U_i\mid U_0^{i-1},\bY}(0\mid U_0^{i-1},\bY)P_{U_i\mid
U_0^{i-1},\bY}(1\mid U_0^{i-1}, \bY)}\right]\\
&=\sum_{u_0^{i-1},\by} P_{U_0^{i-1},\bY}(u_0^{i-1},\by)\\
&\phantom{===}\sqrt{P_{U_i\mid U_0^{i-1},\bY}(0\mid
u_0^{i-1},\by)P_{U_i\mid U_0^{i-1},\bY}(1\mid
u_0^{i-1}, \by)}\\
& = 
\sum_{u_0^{i-1},\by}\sqrt{P_{U_0^{i-1},U_i,\bY}(u_0^{i-1},0,\by)P_{U_0^{i-1},U_i,\bY}(u_0^{i-1},1,\by)}
\\
& =\sum_{u_0^{i-1},\by}\sqrt{\sum_{u_{i+1}^{N-1}}P_{\bU,\bY}((u_0^{i-1},0,u_{i+1}^{N-1}),\by)}\\
&\phantom{========}\sqrt{ \sum_{u_{i+1}^{N-1}}
P_{\bU,\bY}((u_0^{i-1},1,u_{i+1}^{N-1}),\by)}\\
& \stackrel{(a)}{=} \frac{1}{2^N}\sum_{u_0^{i-1},\by} \sqrt{\sum_{u_{i+1}^{N-1}} P_{\bY\mid\bU}(\by \mid
u_0^{i-1},0,u_{i+1}^{N-1})}\\
&\phantom{========}\sqrt{\sum_{ u_{i+1}^{N-1}} P_{\bY\mid\bU}(\by \mid u_0^{i-1},1,u_{i+1}^{N-1})}\\
&= \frac{1}{2}{Z^{(i)}}.
\end{align*}
The equality $(a)$ follows from the fact that $P_{\bU}(\bu) = \frac{1}{2^N}$ for all 
$\bu \in\{0,1\}^N$. 

Assume now that $Z^{(i)} \geq 1- 2\delta_N^2$.
Then
\begin{align*}
&\bE_P\left[\frac12 - \sqrt{P_{U_i\mid U_0^{i-1},\bY }(0\mid
U_0^{i-1},\bY)P_{U_i\mid U_0^{i-1},\bY }(1\mid U_0^{i-1},\bY)}\right]\\
& \phantom{===========================}\leq \delta_N^2.
\end{align*}
Multiplying and dividing the term inside the expectation with 
\begin{align*}
\frac12 +
\sqrt{P_{U_i\mid U_0^{i-1},\bY }(0\mid u_0^{i-1},\by)P_{U_i\mid U_0^{i-1},\bY }(1\mid u_0^{i-1},\by)},
\end{align*} 
and upper bounding this term in the denominator with $1$, we get 
\begin{align*}
\bE_P\left[\frac14 - P_{U_i\mid U_0^{i-1},\bY}(0\mid
U_0^{i-1},\bY)P_{U_i\mid U_0^{i-1},\bY}(1\mid
U_0^{i-1},\bY)\right].
\end{align*}
Now, using the equality $\frac14 - p\bar{p} = (\frac12-p)^2$, we get
\begin{align*}
\bE_P\left[\Big(\frac12 - P_{U_i\mid U_0^{i-1},\bY}(0\mid U_0^{i-1},\bY)\Big)^2\right] \leq
\delta_N^2.
\end{align*}
The result now follows by applying the Cauchy-Schwartz inequality.
\end{proof}

We are now ready to prove Theorem~\ref{the:main}. In order to show that there
exists a polar code which achieves the rate-distortion tradeoff, we show that
the size of the set $F$ can be made arbitrarily close to $N h_2(D)$ while keeping
the penalty term $|F| 2 \delta_N$ arbitrarily small.

{\em Proof of Theorem~\ref{the:main}:}

Let $\beta<\frac12$ be a constant and let $\delta_N = \frac{1}{{2N}}2^{-N^\beta}$. 
Consider a polar code with frozen set $F_N$,
\begin{align*}
F_N = \{i\in \{0,\dots,N-1\}: Z^{(i)} \geq 1-2\delta_N^2\}.
\end{align*}
For $N$ sufficiently large there exists a $\beta' < \frac12$ such that
$2\delta_N^2 > 2^{-N^{\beta'}}$. Theorem~\ref{thm:rateZbar} and equation \eqref{eqn:Z-count}
imply that 
\begin{align}\label{equ:sizeF}
\lim_{N=2^n, n\to \infty} \frac{|F_N|}{N} = h_2(D).
\end{align}
For any $\epsilon >0$ this implies
that for $N$ sufficiently large there exists a set $F_{N}$ such that
\begin{align*}
\frac{|F_N|}{N} \geq h_2(D) - \epsilon. 
\end{align*}
In other words
\begin{align*} 
R_N = 1-\frac{|F_N|}{N} \leq R(D)+\epsilon.
\end{align*}
Finally, from Lemma~\ref{lem:QversusP} we know that
\begin{align}\label{equ:average}
D_N(F_N) \leq D + 2|F_N|\delta_N \leq D + O(2^{-(N^\beta)})
\end{align} 
for any $0< \beta < \frac12$.

Recall that $D_N(F_N)$ is the average of the distortion over all choices of $\tilde{u}_F$.
Since the average distortion fulfills \eqref{equ:average} it follows that 
there must be at least
one choice of $\tilde{u}_{F_N}$ for which
\begin{align*}
D_N(F_N,\tilde{u}_{F_N}) \leq D +O(2^{-(N^\beta)})
\end{align*}
for any $0< \beta < \frac12$.

The complexity of the encoding and decoding algorithms
are of the order $\Theta(N\log(N))$ as shown in \cite{Ari08}.
\qed

\section{Value of Frozen Bits Does Not Matter}\label{sec:gauge}
In the previous sections we have considered $D_N(F)$, the average distortion
if we average over all choices of $\tilde{u}_F$. 
We will now show a stronger result, namely we will show that {\em all} choices
for $\tilde{u}_F$ lead to the same distortion, i.e., $D_N(F,\tilde{u}_F)$ is
independent of $\tilde{u}_F$. This implies that the components
belonging to the frozen set $F$ can be set to any value. A convenient choice is
to set them to $0$. In the
following let $F$ be a fixed set. The results here do not dependent on
the set $F$. 
\begin{lemma}[Gauge Transformation]\label{lem:gaugesymmetry}
Consider the Standard Model introduced in the previous section.
Let $\by,\by' \in\{0,1\}^N$ and let $u_0^{i-1} = {u'}_{0}^{i-1}
\oplus ((\by  \oplus \by') H_n^{-1})_0^{i-1}$. Then 
\begin{align*}
l_i(\by, u_0^{i-1}) = \left\{
\begin{array}{cc}
l_i(\by',{u'}_0^{i-1}), & \text{ if } ((\by \oplus \by') H_n^{-1})_i = 0,\\
1/l_i(\by',{u'}_0^{i-1}), & \text{ if } ((\by \oplus \by') H_n^{-1})_i =1.
\end{array}\right.
\end{align*}
\end{lemma}
\begin{proof}
\begin{align*}
& l_i(\by, u_0^{i-1})  \\
&= \frac{W^{(i)}(\by, u_0^{i-1}\mid 0 )}{W^{(i)}(\by, u_0^{i-1}\mid
1)} \\
&=\frac{\sum_{u_{i+1}^{N-1}} P(\by\mid
u_0^{i-1},0,u_{i+1}^{N-1})}{\sum_{u_{i+1}^{N-1}} P(\by\mid
u_0^{i-1},1,u_{i+1}^{N-1})} \\
& \stackrel{(\ref{equ:channelsymmetry})}{=} \frac{\sum_{u_{i+1}^{N-1}}
P(\bar{y}'\mid (u_0^{i-1},0,u_{i+1}^{N-1})\oplus(\by\oplus\by')H_n^{-1})}{\sum_{u_{i+1}^{N-1}}
P(\bar{y}'\mid (u_0^{i-1},1,u_{i+1}^{N-1})\oplus(\by\oplus\by')H_n^{-1})}\\
&= \frac{\sum_{u_{i+1}^{N-1}}
P(\by'\mid ({u'}_0^{i-1},0\oplus ((\by\oplus\by') H_n^{-1})_i,u_{i+1}^{N-1})}{\sum_{u_{i+1}^{N-1}}
P(\by'\mid ({u'}_0^{i-1},1\oplus ((\by \oplus \by')H_n^{-1})_i,u_{i+1}^{N-1})}\\
&=  \frac{W^{(i)}(\by', {u'}_0^{i-1}\mid 0\oplus ((\by\oplus \by') H_n^{-1})_i)}{W^{(i)}(\by', {u'}_0^{i-1}\mid
1\oplus ((\by \oplus \by')H_n^{-1})_i)}.
\end{align*}
The claim follows by considering the two possible values of $((\by \oplus
\by')H_n^{-1})_i$.
\end{proof}
Recall that the decision process involves randomized rounding on the basis of $l_i$.
Consider at first two tuples $(\by,u_0^{i-1})$ and $(\by',{u'}_0^{i-1})$ 
so that their associated $l_i$ values are equal;  we have seen in the previous
lemma that many such tuples exist. In this case, if both tuples
have access to the same source of randomness, we can couple the two instances
so that they make the same decision on $U_i$. An equivalent statement is true in the
case when the two tuples have the same reliability $|\log(l_i(\by,u_0^{i-1}))|$
but different signs. In this case there is a simple coupling that ensures
that if for the first tuple the decision is lets say $U_i=0$ then for the second tuple
it is $U_i=1$ and vice versa. Hence, if in the sequel we compare two instances of ``compatible" tuples
which have access to the same source of randomness, then we assume exactly this coupling.

\begin{lemma}[Symmetry and Distortion]\label{lem:outputsymmetry}
Consider the Standard model introduced in the previous section. Let $\bar{y}, \bar{y}'
\in \{0,1\}^N$, $F \subseteq \{0, \dots, N-1\}$,  and $\tilde{u}_F, \tilde{u}'_F \in \{0, 1\}^{|F|}$.
If $\tilde{u}_F = \tilde{u}'_F \oplus ((\by \oplus \by')H_n^{-1})_F$, then
under the coupling through a common source of randomness 
$f^{\tilde{u}_F}(\by) =
f^{\tilde{u}'_F}(\by') \oplus ((\by \oplus \by')H_n^{-1})_{F^c}$.
\end{lemma}
\begin{proof}
Let $\bu,\bu'$ be the two $N$ dimensional vectors
generated within the Standard Model.  We use induction.
Fix $0 \leq i \leq N-1$. We assume that for $j < i$, 
$u_j = u'_j \oplus ((\by \oplus \by')H_n^{-1})_j$. This is in particular correct if $i=0$,
which serves as our anchor.

By Lemma \ref{lem:gaugesymmetry} we conclude that under our coupling the 
respective decisions are related as $u_i = u'_i \oplus ((\by \oplus \by')H_n^{-1})_i$ if $i \in F^c$. 
On the other hand, if $i \in F$, then the claim is true by assumption. 
\end{proof}

Let $\bv\in\{0,1\}^{|F|}$ and let $A(\bv) \subset \{0,1\}^N$ denote the coset
\begin{align*} 
A(\bv) = \{\by: (\by H_n^{-1})_F = \bv\}.
\end{align*}
The set of source words $\{0,1\}^N$ can be partitioned as 
\begin{align*}
\{0,1\}^N = \cup_{\bv \in \{0,1\}^{|F|}} A(\bv).
\end{align*}
Note that all the cosets $A(\bv)$ have equal size.

The main result of this section is the following lemma. The lemma implies that the
distortion of SM$(F,\tilde{u}_F)$ is independent of $\tilde{u}_F$.
\begin{lemma}[Independence of Average Distortion w.r.t. $\tilde{u}_F$]\label{lem:symmetryofdist}
Fix $F\subseteq\{0,\dots,N-1\}$. The average distortion $D_N(F,\tilde{u}_F)$ of
the model SM$(F,\tilde{u}_F)$ is independent of the choice of $\tilde{u}_F\in
\{0,1\}^{|F|}$.
\end{lemma}
\begin{proof}
Let $\tilde{u}_F,\tilde{u}'_F \in \{0,1\}^{|F|}$ be two fixed vectors. 
We will now show that $D_N(F,\tilde{u}_F) = D_N(F,\tilde{u}'_F)$. 
Let $\by,\by'$ be two source words such that $\by \in A(\bv)$ and $\by' \in
A(\bv\oplus \tilde{u}_F\oplus \tilde{u}'_F)$, 
i.e., $\tilde{u}'_F = \tilde{u}_F \oplus ((\by\oplus \by')H_n^{-1})_F$. 
Lemma \ref{lem:outputsymmetry} implies
that 
\begin{align*}
f^{\tilde{u}'_F}(\by') = f^{\tilde{u}_F}(\by) \oplus ((\by
\oplus \by')H_n^{-1})_{F^c}.
\end{align*}

This implies that the reconstruction words are related as
\begin{align*}
\hat{f}^{\tilde{u}_F}(f^{\tilde{u}_F}(\by)) =
\hat{f}^{\tilde{u}'_F}(f^{\tilde{u}'_F}(\by')) \oplus (\by\oplus \by')H_n^{-1}.
\end{align*}
Note that $\hat{f}^{\tilde{u}_F}({f}^{\tilde{u}_F}(\by))\oplus \by$ is the
quantization error. Therefore
\begin{align*}
\disto(\by, \hat{f}^{\tilde{u}_F}(f^{\tilde{u}_F}(\by))) = \disto(\by',
\hat{f}^{\tilde{u}_F}(f^{\tilde{u}'_F}(\by'))),
\end{align*}
which further implies 
\begin{align*}
\sum_{\by\in A(\bv)}\disto(\by, \hat{f}^{\tilde{u}_F}(f^{\tilde{u}_F}(\by))) = \sum_{\by\in
A(\bv\oplus\tilde{u}_F\oplus\tilde{u}'_F)}\disto(\by, \hat{f}^{\tilde{u}'_F}(f^{\tilde{u}'_F}(\by))).
\end{align*}
Hence, the average distortions satisfy
\begin{align*}
\sum_{\by} &\frac{1}{2^N}\disto(\by,\hat{f}^{\tilde{u}_F}(f^{\tilde{u}_F}(\by)
)) \\
= & \sum_{\bv \in \{0,1\}^{|F|}} \frac{1}{2^N} \sum_{\by \in
A(\bv)}\disto(\by,\hat{f}^{\tilde{u}_F}(f^{\tilde{u}_F}(\by)))\\
= & \sum_{\bv \in \{0,1\}^{|F|}} \frac{1}{2^N} \sum_{\by \in
A(\bv\oplus\tilde{u}_F\oplus\tilde{u}'_F)}\disto(\by,\hat{f}^{\tilde{u}'_F}(f^{\tilde{u}'_F}(\by)))\\
= & \sum_{\bv \in \{0,1\}^{|F|}} \frac{1}{2^N} \sum_{\by \in
A(\bv)}\disto(\by,\hat{f}^{\tilde{u}'_F}(f^{\tilde{u}'_F}(\by)))\\
= & \sum_{\by} \frac{1}{2^N}\disto(\by,\hat{f}^{\tilde{u}'_F}(f^{\tilde{u}'_F}(\by))).
\end{align*}

As mentioned before, the functions $f^{\tilde{u}_F}$ and $f^{\tilde{u}'_F}$ are not
deterministic and the above equality is valid under the assumption of coupling
with a common source of randomness. Averaging over this common randomness, we
get $D_N(F,\tilde{u}_F) = D_N(F,\tilde{u}'_F)$.
\end{proof}

Let $\quant^{\tilde{u}_F}$ denote the empirical distribution of the quantization
noise, i.e.,
\begin{align*}
\quant^{\tilde{u}_F}(\bar{x}) = \bE[\indicator{\bY \oplus
\hat{f}^{\tilde{u}_F}(f^{\tilde{u}_F}(\bY)) =\bar{x}}],
\end{align*}
where the expectation is over the randomness involved in the source and randomized rounding.
Continuing with the reasoning of the previous lemma, we can indeed show that the
distribution $\quant^{\tilde{u}_F}$ is independent of $\tilde{u}_F$. Combining
this with Lemma~\ref{lem:totalvarbnd}, we can bound the distance
between $\quant^{\tilde{u}_F}$ and an i.i.d. Ber$(D)$ noise. This
will be useful in settings which involve both channel and source coding, like the
Wyner-Ziv problem, where it is necessary to show that the quantization noise is
close to a Bernoulli random variable.

\begin{lemma}[Distribution of the Quantization Error]\label{lem:quantdist}
Let the frozen set $F$ be 
\begin{align*}
F = \{i : Z^{(i)} \geq 1-2\delta_N^2\}.
\end{align*}
Then for $\tilde{u}_F$ fixed, 
\begin{align*}
\sum_{\bar{x}}|\quant^{\tilde{u}_F}(\bar{x}) - \prod_{i}W(x_i\mid0)| \leq
2|F|\delta_N.
\end{align*}
\end{lemma}
\begin{proof}
Recall that $P_{\bX\mid\bY}(\bx\mid \by)=\prod_i W(x_i \mid y_i)$. 
Let $\bv \in \{0,1\}^{|F|}$ be a fixed vector. Consider a vector $\by \in
A(\bv)$ and set $\by'=\bar{0}$. Lemma \ref{lem:outputsymmetry} implies that
$f^{\tilde{u}_F}(\by) = f^{\tilde{u}_F \oplus \bv}(\bar{0}) \oplus (\by
H_n^{-1})_{F^c}$. Therefore, 
\begin{align*}
\by \oplus \hat{f}^{\tilde{u}_F}(f^{\tilde{u}_F}(\by)) =
\bar{0} \oplus \hat{f}^{\tilde{u}_F\oplus \bv}(f^{\tilde{u}_F\oplus\bv}(\bar{0})).
\end{align*}
This implies that all vectors belonging to $A(\bv)$ have the same quantization
error and this error is equal to the error incurred by the all-zero word
when the frozen bits are set to $\tilde{u}_F \oplus \bv$.

Moreover, the uniform distribution of the source induces a uniform distribution on the 
sets $A(\bv)$ where $\bv \in \{0, 1\}^{|F|}$. Therefore, the distribution of the
quantization error $\quant^{\tilde{u}_F}$ is the same as first picking the
coset uniformly at random, i.e., the bits $\tilde{u}_F$, and then generating the
error $\bar{x}$ according to $\bar{x} = \hat{f}^{\tilde{u}_F}(f^{\tilde{u}_F}(\bar{0}))$.
The distribution of the vector $\bu$ where $\bu = \bar{x}H_n^{-1}$ is indeed 
the distribution $Q$ defined in \eqref{eqn:distQ}.
Recall that in the distribution $P_{\bU,\bX,\bY}$, $\bU$ and $\bX$ are related
as $\bU = \bX H_n^{-1}$.
Therefore, the distribution induced by $W(\bx\mid\by)$ on $\bU$ is $P_{\bU\mid
\bY}$. Since multiplication with $H_n^{-1}$ is a one-to-one mapping, the total
variation distance can be bounded as 
\begin{align*}
\sum_{\bar{x}}|\quant^{\tilde{u}_F}(\bar{x})-\prod_i W(\bar{x}\mid\bar{0})| 
&= \sum_{\bu} |Q(\bu\mid \bar{0}) - P_{\bU\mid\bY}(\bu\mid\bar{0})| \\
& \stackrel{(a)}{\leq} 2|F|\delta_N.
\end{align*} 
The inequality $(a)$ follows from Lemma~\ref{lem:totalvarbnd} and
Lemma~\ref{lem:Zdelta}.
\end{proof}

\section{Beyond Source Coding}\label{sec:extensions}
Polar codes were originally defined in the context of channel coding in
\cite{Ari08}, where it was shown that they achieve the capacity of symmetric
B-DMCs. Now we have seen that polar codes achieve the rate-distortion tradeoff for lossy
compression of a BSS.
The natural question to ask next is whether these codes are suitable for
problems that involve both quantization as well as error correction.

Perhaps the two most prominent examples are the source coding problem
with side information (Wyner-Ziv problem \cite{WyZ76}) as well as the channel coding problem
with side information (Gelfand-Pinsker problem \cite{GeP83}). 
As discussed in \cite{ZSE02}, nested linear codes are required to tackle these problems. 
Polar codes are equipped with such a nested structure and are, hence, natural
candidates for these problems. We will show that, by taking advantage of this structure,
one can construct polar codes that are optimal in both settings 
(for the binary versions of these problems). 
Hence, polar codes provide the first provably optimal low-complexity solution.

In \cite{WaM07} the authors constructed MN codes which have the
required nested structure. They show that these codes achieve the optimum
performance under MAP decoding. How these codes perform
under low complexity message-passing algorithms is still an open problem. 
Trellis and turbo based codes were considered in
\cite{CPR03,PrR03,LXG03,YSXZ08} for
the Wyner-Ziv problem. It was empirically shown that they achieve good
performance with low complexity message-passing algorithms.
A similar combination was considered in \cite{CPR01,ErB05,SLSX05}
for the Gelfand-Pinsker problem. Again, empirical results close to the optimum
performance were obtained.

We end this section by applying polar codes to a multi-terminal setup.
One such scenario was considered in \cite{HKU09}, where
it was shown that polar codes are optimal for 
lossless compression of a correlated binary source (the Slepian-Wolf problem \cite{SlW73}).
The result follows by mapping the lossless source compression task to 
a channel coding problem. 

Here we consider another multi-terminal setup known as
the one helper problem \cite{Wyn73}. 
This problem involves channel coding at one terminal and source coding at the
other. We again show that polar codes achieve optimal performance under low-complexity
encoding and decoding algorithms.

\subsection{Binary Wyner-Ziv Problem}
Let $Y$ be a BSS and let the decoder have access to a random variable
$Y'$. This random variable is usually called the {\em side information}.
We assume that $Y'$ is correlated to $Y$ as $Y' = Y + Z$, where $Z$
is a Ber$(p)$ random variable. The task of the encoder is to compress
the source $Y$, call the result $X$, such that a decoder with access to $(Y', X)$ can reconstruct
the source to within a distortion $D$.

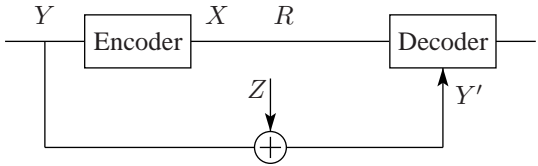
\begin{figure}[ht]
\centering
\input{ps/wynerziv}
\caption{The side information $Y'$ is available at the decoder. The decoder
wants to reconstruct the source $Y$ to within a distortion $D$ given $X$.}\label{fig:wynerziv}
\end{figure}

Wyner and Ziv \cite{WyZ76} have shown that the rate-distortion curve for this problem
is given by 
\begin{align*}
{l.c.e.} \Big\{(R_{\WZ}(D),D), (0,p)\Big\},
\end{align*}
where $R_{\WZ}(D) = h_2(D\ast p) - h_2(D)$, $l.c.e.$
denotes the {\em lower convex envelope}, and $D \ast p = D(1-p) + p(1-D)$.
Here we focus on achieving the rates of the form $R_{\WZ}(D)$. The remaining
rates can be achieved by appropriate time-sharing with the pair  $(0,p)$.

The proof is based on the following nested code construction. 
Let $\code_s$ denote the polar code defined by the frozen set $F_s$ with
the frozen bits $\bu_{F_s}$ set to $0$. 
Let $\code_c(\bv)$ denote the code defined by the frozen set
$F_c \supset F_s$ with the frozen bits $\bu_{F_s}$ set to $0$ and
$\bu_{F_c\backslash F_s} = \bv$. 
This implies that the code
$\code_s$ can be partitioned as $\code_s = \cup_{\bv}\code_c(\bv)$. 

The code $\code_s$ is designed to be a good source
code for distortion $D$ and for each $\bv$ the code $\code_c(\bv)$ is designed to be a good 
channel code for the BSC$(D\ast p)$. 

The encoder compresses the source vector $\bY$ to a vector $\bU_{F_s^c}$ through the map
$\bU_{F_s^c} = f^{\bar{0}}(\bY)$. The reconstruction vector $\bX$ is given by
$\bX = \hat{f}^{\bar{0}}(f^{\bar{0}}(\bar{Y}))$.
Since the code $\code_s$ is a good
source code, the quantization error $\bY \oplus \bX$ is close
to a Ber$(D)$ vector (see Lemma~\ref{lem:quantdist}). This implies that the vector 
$\bY'$ which is available at
the decoder is statistically equivalent to the output of a
BSC$(D\ast p)$ when the input is $\bX$.
The encoder transmits the vector $\bV = \bU_{F_c\backslash F_s}$ to the
decoder. This informs the decoder of the code $\code_c(\bV)$ which is used.
Since this code $\code_c(\bV)$ is designed for the BSC$(D\ast p)$, the decoder
can with high probability determine $\bX$ given $\bY'$. 
By construction, $\bX$ represents $\bY$ with distortion roughly $D$
as desired.

\begin{theorem}[Optimality for the Wyner-Ziv Problem]\label{the:mainWZ}
Let $Y$ be a BSS and $Y'$ be a Bernoulli random variable correlated to $Y$ as
$Y'=Y\oplus Z$, where $Z\sim \text{Ber}(p)$. 
Fix the {\em design} distortion $D$, $0 < D < \frac12$.
For any rate $R > h_2(D\ast p)-h_2(D)$ and any $0 < \beta < \frac12$,
there exists a sequence of nested polar codes of length $N$ with rates $R_N <R$  
so that under SC encoding using randomized rounding at the
encoder and SC decoding at the decoder, they
achieve expected distortion $D_N$ satisfying
\begin{align*}
D_N & \leq D + O(2^{-(N^\beta)}),
\end{align*}
and the block error probability satisfying
\begin{align*}
\Pb_N \leq O(2^{-(N^\beta)}).
\end{align*}
The encoding as well as decoding complexity of these codes is $\Theta(N \log(N))$.
\end{theorem}
\begin{proof}
Let $\epsilon > 0$ and $0 < \beta < \frac12$ be some constants. Let $Z^{(i)}(q)$
denote the $Z^{(i)}$s computed with $W$ set to BSC$(q)$. 
Let $\delta_N = \frac{1}{N}2^{-(N^{\beta})}$.
Let $F_s$ and $F_c$ denote the sets
\begin{align*}
F_s & = \{i: Z^{(i)}(D) \geq 1-\delta_N^2\},\\
F_c & = \{i: Z^{(i)}(D\ast p) \geq \delta_N\}.
\end{align*} 
Theorem~\ref{thm:rateZbar} implies that for $N$ sufficiently large 
\begin{align*}
\frac{|F_s|}{N}\geq h_2(D) - \frac\epsilon2.
\end{align*}
Similarly, Theorem~\ref{thm:rateZ} implies that for $N$ sufficiently large 
\begin{align*}
\frac{|F_c|}{N}\leq h_2(D\ast p) + \frac\epsilon2.
\end{align*}
The degradation of BSC$(D\ast p)$ with respect to BSC$(D)$ implies
that $F_s \subset F_c$. 

The bits $F_s$ are fixed to $0$. This is known both to the encoder and the decoder.
A source vector $\by$ is mapped to $\bu_{F_s^c} = f^{\bar{0}}(\by)$ as shown in the Standard Model.
Therefore the average distortion $D_N$ is bounded as 
\begin{align*}
D_N \leq D + 2|F_s|\delta_N \leq D + O(2^{-(N^\beta)}).
\end{align*}
The encoder transmits the vector $\bu_{F_c\backslash F_s}$ to the
decoder. The required rate is 
\begin{align*}
R_N = \frac{|F_c|-|F_s|}{N} \leq h_2(D\ast p) - h_2(p) + \epsilon.
\end{align*} 

It remains to show that at the decoder the block error probability incurred in
decoding $\bX$ is $O(2^{-(N^\beta)})$.

Let $\bar{E}$ denote the quantization error,  $\bar{E} = \bY\oplus\bX$.
The information available at the decoder ($\bY'$) can be expressed as,
\begin{align*}
\bY' = \bX \oplus \bar{E}\oplus \bZ.
\end{align*}

Consider the code $\code_c(\bv)$ for a given $\bv$ and transmission over the 
BSC$(D \ast p)$. Let $\mathcal{E}\subseteq\{0,1\}^N$ denote the set of noise vectors 
of the channel which result in a decoding error under SC decoding.
By the equivalent of Lemma~\ref{lem:symmetryofdist} for the channel coding case,
this set does not depend on $\bv$. 

The block error probability of our scheme can then be expressed as
\begin{align*}
\Pb_N &= \bE[\indicator{\bar{E} \oplus \bar{Z} \in \mathcal{E}}].
\end{align*}

The exact distribution of the quantization error is not known, but
Lemma~\ref{lem:quantdist} provides a bound on the total variation distance
between this distribution and an i.i.d. Ber$(D)$ distribution. Let $\bar{B}$
denote an i.i.d. Ber$(D)$ vector.
Let $P_{\bar{E}}$ and $P_{\bar{B}}$ denote the distribution of $\bar{E}$ and
$\bar{B}$ respectively. Then
\begin{align}\label{eqn:TVQP}
\sum_{\be}\vert P_{\bar{E}}(\be) - P_{\bar{B}}(\be)\vert  \leq 2|F_s|\delta_N
\leq O(2^{-(N^\beta)}).
\end{align}

Let $\Pr(\bar{B},\bar{E})$ denote the so-called {\em optimal coupling} between $\bar{E}$ and
$\bar{B}$. I.e., a joint distribution of $\bar{E}$ and $\bar{B}$
with marginals equal to $P_{\bar{E}}$ and $P_{\bar{B}}$, and satisfying 
\begin{align}\label{eqn:couplingTV}
\Pr(\bar{E} \neq \bar{B}) = \sum_{\be}\vert P_{\bar{E}}(\be) -
P_{\bar{B}}(\be)\vert.
\end{align}

It is known \cite{Aldbook} that such a coupling exists.
Let $\bar{E}$ and $\bar{B}$ be generated according to $\Pr(\cdot,\cdot)$. 
Then, the block error probability can be expanded as 
\begin{align*}
\Pb_N & = \bE[\indicator{\bar{E} \oplus \bar{Z} \in \mathcal{E}}
\indicator{\bar{E}=\bar{B}}]
+ \bE[\indicator{\bar{E} \oplus \bar{Z} \in
\mathcal{E}}\indicator{\bar{E}\neq \bar{B}}]\\
&\leq   \bE[\indicator{\bar{B} \oplus \bar{Z} \in \mathcal{E}}]+
\bE[\indicator{\bar{E}\neq \bar{B}}]
\end{align*}
The first term in the sum refers to the block error probability for the
BSC$(D\ast p)$, which can be bounded as 
\begin{align}\label{eqn:pBDp}
\bE[\indicator{\bar{B}\oplus \bZ \in \mathcal{E}}] \leq \sum_{i\in F_c} Z^{(i)}(D\ast p)
\leq O(2^{-(N^\beta)}).
\end{align}
Using \eqref{eqn:TVQP}, \eqref{eqn:couplingTV} and \eqref{eqn:pBDp} we get
\begin{align*}
P_N^B \leq O(2^{-(N^\beta)}).
\end{align*}

\end{proof}

\subsection{Binary Gelfand-Pinsker Problem}
Let $S$ denote a symmetric Bernoulli random variable. Consider a channel with
state $S$ given by
\begin{align*}
Y = X \oplus S \oplus Z,
\end{align*}
where $Z$ is a Ber$(p)$ random variable. The state $S$ is known to the encoder
a-causally and not known to the decoder. The output of the encoder is
constrained to satisfy $\bE[X]\leq D$, i.e., on average the fraction of 1s it
can transmit is bounded by $D$. This is similar to the power constraint in the
continuous case. The task of the encoder is to transmit a message $M$ to the
decoder with vanishing error probability under the above mentioned input
constraint.

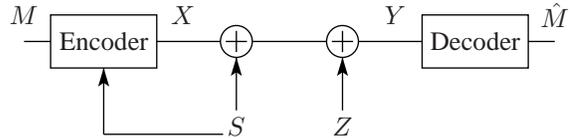
\begin{figure}[ht]
\centering
\input{ps/gelfandpinsker}
\caption{The state $S$ is known to the encoder in advance. The weight of the
input $X$ is constrained to $\bE[X]\leq D$.}\label{fig:gelfandpinsker}
\end{figure}

In \cite{BCW03}, it was shown that the achievable rate, weight pairs for this channel are given by
\begin{align*}
u.c.e. \Big\{(R_{\GP}(D),D),(0,0)\Big\},
\end{align*}
where $R_{\GP}(D) = h_2(D)-h_2(p)$, and $u.c.e$ denotes the upper convex envelope.

Similar to the Wyner-Ziv problem, we need a nested code for this problem.
However, they differ in the sense that the role of the channel and source codes
are reversed.
 
Let $\code_c$ denote the polar code defined by the frozen set $F_c$ with
frozen bits $\bu_{F_c}$ set to $0$. Let $\code_s(\bv)$ denote the code defined by
the frozen set $F_s \supset F_c$, with the frozen bits $\bu_{F_c}$ set to $0$
and $\bu_{F_s\backslash F_c} = \bv$.
The code $\code_c$ is designed to be a good channel code for the BSC$(p)$
and the codes $\code_s(\bv)$ are designed to be good source codes for
distortion $D$. 
This implies that the code
$\code_c$ can be partitioned into $\code_s(\bv)$ for $\bv \in \{0,1\}^{F_s\backslash
F_c}$, i.e., $\code_c = \cup_{\bv}\code_s(\bv)$.

The frozen bits $\bV = \bU_{F_s\backslash F_c}$ are determined by the message $M$ that
is transmitted. The encoder
compresses the state vector $\bar{S}$ to a vector
$\bU_{F_s^c}$ through the map $\bU_{F_s^c} = f^{\bar{U}_{F_s}}(\bar{S})$. Let
$\bar{S}'$ be the reconstruction vector
$\bar{S}'=\hat{f}^{\bar{U}_{F_s}}(f^{\bar{U}_{F_s}}(\bar{S}))$.
The encoder sends the vector $\bX = \bar{S}\oplus\bar{S}'$ through the channel.
Since the codes $\code_s(\bV)$ are good source codes, the expected distortion
$\frac{1}{N}\bE[\disto(\bar{S} ,\bar{S}')]$ (hence the average weight of $\bX$) 
is close to $D$ (see Lemma~\ref{lem:symmetryofdist}). Since the code $\code_c$
is designed for the BSC$(p)$, the decoder will succeed in decoding the codeword
$\bar{S}\oplus\bX = \bar{S}'$ (hence the message $\bV$) with high probability. 

Here we focus on achieving the rates of the form $R_{\GP}(D)$. The
remaining rates can be achieved by appropriate time-sharing with the pair
$(0,0)$.

\begin{theorem}
[Optimality for the Gelfand-Pinsker Problem]\label{the:mainGP}
Let $S$ be a symmetric Bernoulli random variable. Fix $D$, $0 < D < \frac12$.
For any rate $R < h_2(D)-h_2(p)$ and any $0 < \beta < \frac12$,
there exists a sequence of polar codes of length $N$
so that under SC encoding using randomized rounding at the
encoder and SC decoding at the decoder, the achievable rate
satisfies
\begin{align*}
R_N > R,
\end{align*}
with the expected weight of $X$, $D_N$, satisfying
\begin{align*}
D_N & \leq D + O(2^{-(N^\beta)}). 
\end{align*}
and the block error probability satisfying
\begin{align*}
\Pb_N \leq O(2^{-(N^\beta)}).
\end{align*}
The encoding as well as decoding complexity of these codes is $\Theta(N \log(N))$.
\end{theorem}
\begin{proof}
Let $\epsilon > 0$ and $0 < \beta < \frac12$ be some constants. Let $Z^{(i)}(q)$
denote the $Z^{(i)}$s computed with $W$ set to BSC$(q)$. Let $\delta_N =
\frac{1}{N}2^{-(N^{\beta})}$.
Let $F_s$ and $F_c$ denote the sets
\begin{align}
F_s & = \{i: Z^{(i)}(D) \geq 1-\delta_N^2\},\label{eqn:frozenGPs}
 \\
F_c & = \{i: Z^{(i)}(p) \geq \delta_N\}.\label{eqn:frozenGPc}
\end{align} 
Theorem~\ref{thm:rateZbar} implies that for $N$ sufficiently large 
\begin{align*}
\frac{|F_s|}{N}\geq h_2(D) - \frac\epsilon2.
\end{align*}
Similarly, Theorem~\ref{thm:rateZ} implies that for $N$ sufficiently large 
\begin{align*}
\frac{|F_c|}{N}\leq h_2(p) + \frac\epsilon2.
\end{align*}
The degradation of BSC$(D)$ with respect to BSC$(p)$ implies
that $F_c \subset F_s$. 
The vector $\bu_{F_{s}\backslash F_c}$ is defined by the message that is
transmitted.
Therefore, the rate of transmission is 
\begin{align*}
\frac{|F_s| -|F_c|}{N} \geq h_2(D) - h_2(p) -\epsilon.
\end{align*}

The vector $\bar{S}$ is compressed using the source code with frozen set $F_s$.
The frozen vector $\bu_{F_s}$ is defined in two stages. The subvector
$\bu_{F_c}$ is fixed to $0$ and is known to both the transmitter and the
receiver.
The subvector $\bu_{F_s\backslash F_c}$ is defined by the message being transmitted. 

Let $\bar{S}$ be mapped to a reconstruction vector $\bar{S}'$.
Lemma~\ref{lem:symmetryofdist} implies that the average distortion 
of the Standard Model is independent of the value of the frozen bits. This implies 
\begin{align*}
\bE[\bar{S}\oplus \bar{S}'] \leq 
D+2|F_s|\delta_N \leq D+O(2^{-(N^\beta)}).
\end{align*}
Therefore, a transmitter which sends $\bX = \bar{S} \oplus \bar{S}'$ will on
average be using $D+O(2^{-(N^\beta)})$ fraction of $1$s.
The received vector is given by 
\begin{align*}
\bY = \bar{X} \oplus \bar{S} \oplus\bar{Z} =
\bar{S}'\oplus\bar{Z}. 
\end{align*}
The vector $\bar{S}'$ is a codeword of $\code_c$, the code designed for the
BSC$(p)$ (see \eqref{eqn:frozenGPc}).
Therefore, the block error probability of the SC decoder in decoding $\bar{S}'$
(and hence $\bar{V}$) is bounded as 
\begin{align*}
\Pb_N \leq \sum_{i\in F_c^c} Z^{(i)}(p) \leq O(2^{-(N^{\beta})}).
\end{align*}
\end{proof}

\subsection{Storage in Memory With Defects}
Let us briefly discuss another standard problem in the literature that fits within the
Gelfand-Pinsker framework but where the state is non-binary. 
Consider the problem of storing data on a computer
memory with defects and noise, explored in \cite{HeG83} and \cite{Tsy75}.
Each memory cell can be in three possible states, say $\{0,1,\ast\}$.
The state $S =0 \;(1)$ means that the value of the cell is stuck at $0\; (1)$ and
$S=\ast$ means that the value of the cell is flipped with probability $D$. Let
the probability distribution of $S$ be 
\begin{align*}
\Pr(S=0) = \Pr(S=1) = p/2,\;\;\; \Pr(S=\ast) = 1-p.
\end{align*}
The optimal storage capacity when the whole state realization is known in
advance only to the encoder is $(1-p)(1-h_2(D))$.

\begin{theorem}
[Optimality for the Storage Problem]\label{the:mainstorage}
For any rate $R < (1-p)(1-h_2(D))$ and any $0 < \beta < \frac12$,
there exists a sequence of polar codes of length $N$
so that under SC encoding using randomized rounding at the
encoder and SC decoding at the decoder, the achievable rate
satisfies
\begin{align*}
R_N > R,
\end{align*}
and the block error probability satisfying
\begin{align*}
\Pb_N \leq O(2^{-(N^\beta)}).
\end{align*}
The encoding as well as decoding complexity of these codes is $\Theta(N \log(N))$.
\end{theorem}

The problem can be framed as a Gelfand-Pinsker setup with state $S \in \{0,1,\ast\}$. 
As seen before, the nested construction for such a problem consists of a good source
code which partitions into cosets of a good channel code. We still need to
define what the corresponding source and coding problems are. 

{\em Source Code: } The source code is designed to compress the ternary source $S$ to
the binary alphabet $\{0,1\}$ with design distortion $D$. The distortion function
is $\disto(0,1)=1,\;\disto(\ast,1)=\disto(\ast,0) = 0,$. The test channel for
this problem is a binary symmetric erasure channel (BSEC) shown in Figure~\ref{fig:BEStest}. The
compression of this source is explained in Section~\ref{sec:discussion}. Let
$Z^{(i)}(p,D)$ denote the Bhattacharyya values of BSEC$(p,D)$ defined in
Figure~\ref{fig:BEStest}. The frozen set $F_s$ is defined as 
\begin{align*}
F_s & = \{i: Z^{(i)}(p,D) \geq 1-\delta_N^2\}.
\end{align*} 
The rate distortion function for this problem is given by $p(1-h_2(D))$. Therefore, for
sufficiently large $N$, $|F_s|/N$ can be made arbitrarily close to
$1-p(1-h_2(D))$. 

{\em Channel code:} The channel code is designed for BSC$(D)$. The frozen set $F_c$ is
defined as 
\begin{align*}
F_c & = \{i: Z^{(i)}(D) \geq \delta_N\}.
\end{align*} 
Therefore, for sufficiently large $N$, $|F_c|/N$ can be made arbitrarily close to
$h_2(D)$. Degradation of BSEC$(p,D)$ with respect to BSC$(D)$ implies $F_c \subseteq F_s$.

{\em Encoding:} The frozen bits $\bU_{F_c}$ is fixed to $\bar{0}$. The
vector $\bU_{F_s\backslash F_c}$ is defined by the message to be stored.
Therefore, the achievable rate is 
\begin{align*}
R_N=\frac{|F_s| - |F_c|}{N} \geq (1-p)(1-h_2(D)) -\epsilon 
\end{align*}
for any $\epsilon > 0$. Compress the source sequence using the function
$f^{\bU_{F_s}}(\bar{S})$ and store the reconstruction vector $\bX =
f^{\bU_{F_s}}(f^{\bU_{F_s}}(\bar{S}))$ in the memory.
As shown in the Wyner-Ziv setting, the quantization noise is close to Ber$(D)$ for the
stuck bits. Therefore, a fraction $D$ of the stuck bits differ from $\bX$.

{\em Decoding:} When the decoder reads the memory, the stuck bits are read
as it is and the remaining bits are flipped with probability $D$. This is
equivalent to seeing $\bX$ through a channel BSC$(D)$. Since the channel
code is defined for BSC$(D)$, the decoding will be successful with high
probability and the message $\bU_{F_s\backslash F_c}$ will be recovered.

\subsection{One Helper Problem}
Let $Y$ be a BSS and let $Y'$ be correlated to $Y$ as $Y'=Y\oplus Z$, where $Z$
is a Ber$(p)$ random variable. The encoder has access to $Y$ and the helper has access to
$Y'$. The aim of the decoder is to reconstruct $Y$ successfully. As the name suggests, the
role of the helper is to assist the decoder in recovering $Y$. 
This problem was considered by Wyner in \cite{Wyn73}.
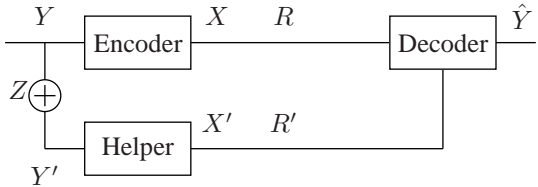
\begin{figure}[ht]
\centering
\input{ps/onehelper}
\caption{The helper transmits quantized version of $Y'$. The decoder uses the
information from the helper to decode $Y$ reliably.}\label{fig:onehelper}
\end{figure}

Let the rates used by the encoder and the helper be $R$ and $R'$ respectively.
Wyner \cite{Wyn73} showed that the required rates $R,R'$ must
satisfy
\begin{align*}
R>h_2(D\ast p),\;\;\; R' > 1-h_2(D),
\end{align*}
for some $D\in[0,1/2]$.
\begin{theorem}
[Optimality for the One Helper Problem]\label{the:mainSCSI}
Let $Y$ be a BSS and $Y'$ be a Bernoulli random variable correlated to $Y$
as $Y' = Y\oplus Z$, where $Z\sim \text{Ber}(p)$. Fix the {\em design}
distortion $D$, $0 < D < \frac12$.
For any rate pair $R > h_2(D\ast p), R'> 1-h_2(D)$ and any $0 < \beta < \frac12$,
there exist sequences of polar codes of length $N$ with rates $R_{N} <R$ and
$R'_{N} < R'$
so that under syndrome computation at the encoder, 
SC encoding using randomized rounding at the
helper and SC decoding at the decoder, they
achieve the block error probability satisfying
\begin{align*}
\Pb_N \leq O(2^{-(N^\beta)}).
\end{align*}
The encoding as well as decoding complexity of these codes is $\Theta(N \log(N))$.
\end{theorem}

For this problem, we require a good channel code at the encoder
and a good source code at the helper. We will explain the code
construction here. The rest of the proof is similar to the previous setups.

{\em Encoding: } The helper quantizes the vector $\bY'$ to $\bX'$ with a design
distortion $D$. This compression can be achieved with rates arbitrarily close to
$1-h_2(D)$.

The encoder designs a code for the BSC$(D\ast p)$. Let $F$ denote the frozen
set.
The encoder computes the syndrome $\bU_F = (\bY H_n^{-1})_F$
and transmits it to the decoder. The rate involved in such an operation is $R =
|F|/N$. Since the fraction $|F|/N$ can be made arbitrarily close
to $h_2(D\ast p)$, the rate $R$ will approach $h_2(D\ast p)$.

{\em Decoding: }
The decoder first reconstructs the vector $\bX'$. 
The remaining task is to decode the codeword $\bY$ from the observation $\bX'$.
As shown in the Wyner-Ziv setting, the quantization noise $\bY\oplus \bX'$
is very ``close" to Ber$(D\ast p)$. 
Note that the decoder knows the syndrome $\bU_F = (\bY H_n^{-1})_F$, where the
frozen set $F$ is designed for the BSC$(D\ast p)$. Therefore, the task of the
decoder is to recover the codeword of a code designed for BSC$(D\ast p)$ when the
noise is close to Ber$(D\ast p)$. Hence the decoder will succeed with high probability.

\section{Complexity Versus Gap}
We have seen that polar codes under SC encoding achieve the
rate-distortion bound when the blocklength $N$ tends to infinity.
It is also well-known that the encoding as well as decoding complexity
grows like $\Theta(N \log(N))$. How does the complexity grow as a function
of the gap to the rate-distortion bound? This is a much more subtle
question.

To see what is involved in being able to answer this question, consider the Bhattacharyya constants
$Z^{(i)}$ defined in \eqref{eqn:channelZ}. 
Let $\tilde{Z}^{(i)}$ denote a re-ordering of these values in an increasing order, i.e., 
 $\tilde{Z}^{(i)} \leq \tilde{Z}^{(i+1)}$, $i=0, \dots, N-2$.
Define
\begin{align*}
m_N^{(i)} & = \sum_{j=0}^{i-1} \tilde{Z}^{(i)}, \\
M_N^{(i)} & = \sum_{j=N-i}^{N-1} \sqrt{2(1-\tilde{Z}^{(i)})}.
\end{align*}

For the binary erasure channel there is a simple recursion to compute
the $\{Z^{(i)}\}$ as shown in \cite{Ari08}.  For general channels
the computation of these constants is more involved but the basic
principle is the same.

For the channel coding problem we then get an upper bound on the block error
probability $\Pb_N$
as a function the rate $R$ of the form
\begin{align*}
(\Pb_N, R) = (m_N^{(i)}, \frac{i}{N}).
\end{align*} 
On the other hand, for the source coding problem,
we get an upper bound on the distortion $D_N$ as a function of the rate of the form
\begin{align*}
(D_N, R) = ( D+M_N^{(i)}, \frac{i}{N}).
\end{align*}
Now, if we knew the distribution of $Z^{(i)}$s it would allow us
to determine the rate-distortion performance achievable for this
coding scheme for any given length.  The complexity per bit is always
$\Theta(\log N)$.

Unfortunately, the computation of the quantities $m_N^{(i)}$ and
$M_N^{(i)}$ is likely to be a challenging problem. Therefore, we
ask a simpler question that we can answer with the estimates we
currently have about the $\{Z^{(i)}\}$.

Let $R = R(D)+ \delta$,  where $\delta >0$.  How does the complexity
per bit scale with respect to the gap between the actual (expected)
distortion $D_N$ and the design distortion $D$? Let us answer this
question for the various low-complexity schemes that have been
proposed to date.

{\em Trellis Codes:}
In \cite{ViO74} it was shown that, using trellis codes and Viterbi
decoding, the average distortion scales like $D+O(2^{-K E(R)})$,
where $E(R) > 0$ for $\delta >0$ and $K$ is the constraint length.
The complexity of the decoding algorithm is $\Theta(2^K N)$. Therefore,
the complexity per bit in terms of the
gap is given by $O(2^{(\log\frac1g)})$.

{\em Low Density Codes:}
In \cite{CiM05} it was shown that under optimum encoding the gap
is $O(\sqrt K 2^{-K\Delta})$, for some $\Delta > 0$, where $K$ is
the average degree of the parity check node. Assuming that using
BID we can achieve this distortion, the complexity is given by
$\Theta(2^K N)$. Therefore, the complexity per bit in terms of the gap
is given by $O(2^{(\log \frac1g)})$.

{\em Polar Codes:}
For polar codes, the complexity is $\Theta(N \log N)$ and the gap is
$O(2^{-(N^\beta)})$ for any $\beta < \frac12$. Therefore, the complexity per bit in terms of
the gap is $O(\frac1{\beta} \log\log\frac1g)$.
This is considerably lower than for the two previous schemes.

\section{Discussion and Future Work}\label{sec:discussion}
We have considered the lossy source coding problem for the BSS 
and the Hamming distortion. The reconstruction alphabet in this case 
is also binary and the test channel ``$W$'' is a BSC. 

Consider the slightly more general scenario of a $q$-ary source
with a binary {\em reconstruction} alphabet. Assume further
that the test channel, call it $W$, is such that the marginal induced by the
source distribution on the reconstruction alphabet is uniform. 

\begin{example}[Binary Erasure Source]
Let the source alphabet be $\{0,1,\ast\}$. Let $S$ denote the source variable
with distribution
\begin{align*}
\Pr(S=1) = \Pr(S=0) = p/2,\; \;\;\Pr(S=\ast) = 1-p.
\end{align*}
Let the distortion function be 
\begin{align}\label{eqn:dist3}
\disto(0,\ast) = \disto(1,\ast) = 0,\;\;\disto(0,1) = 1.
\end{align} For a design distortion $D$, the test channel 
$W:\{0,1\}\to \{0,1,\ast\}$ is shown in Figure~\ref{fig:BEStest}.
Note that the distribution induced on the
input of the channel is uniform.
\begin{figure}
\begin{center}
\begin{picture}(85,100)(0,-10)
\thicklines
\put(0,0){\circle*{4}}
\put(0,0){\line(1,0){80}}
\put(0,0){\line(1,1){80}}
\put(0,0){\line(2,1){80}}
\put(0,80){\circle*{4}}
\put(0,80){\line(1,0){80}}
\put(0,80){\line(2,-1){80}}
\put(0,80){\line(1,-1){80}}
\put(80,0){\circle*{4}}
\put(80,40){\circle*{4}}
\put(80,80){\circle*{4}}
\put(-7,0){\makebox(0,0)[c]{$1$}}
\put(-7,80){\makebox(0,0)[c]{$0$}}
\put(87,0){\makebox(0,0)[c]{$1$}}
\put(87,40){\makebox(0,0)[c]{$\ast$}}
\put(87,80){\makebox(0,0)[c]{$0$}}
\put(40,85){{\makebox(0,0)[b]{$p(1-D)$}}}
\put(40,-5){{\makebox(0,0)[t]{$p(1-D)$}}}
\put(40,62){\makebox(0,0)[b]{$\bar{p}$}}
\put(40,18){\makebox(0,0)[t]{$\bar{p}$}}
\put(22,52){\makebox(0,0)[r]{$pD$}}
\put(22,28){\makebox(0,0)[r]{$pD$}}
\end{picture}
\caption{\label{fig:BEStest}The test channel for the binary erasure source.}
\end{center}
\end{figure}
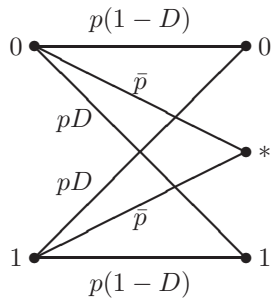
\end{example}

For this setup one can obtain results mirroring
Theorem~\ref{the:main}. More precisely, one can show that the optimum
rate-distortion tradeoff can again be achieved by polar codes
together with SC encoding and randomized-rounding. The proof is analogous
to the proof of Theorem~\ref{the:main}. The only change in the proof consists of replacing
the BSC$(D)$ with the appropriate test channel $W$.
This is the source coding equivalent of Ar\i kan's channel coding result \cite{Ari08}, where
it was shown that polar codes achieve the symmetric mutual information $I(W)$
for any B-DMC.

A further important generalization is the compression of {\em non-symmetric} sources.
Let us explain the involved issues by means of the channel coding problem.
Consider an asymmetric B-DMC, e.g., the $Z$-channel. Due to the asymmetry, the capacity-achieving input distribution
is in general not the uniform one. To be concrete, assume that it is $(p(0)=\frac13, p(1)=\frac23)$.
This causes problems for any scheme which employs linear codes, since linear codes
induce uniform marginals. To get around this problem, ``augment" the channel to a $q$-ary input channel by
duplicating some of the inputs.  For our running example, Figure~\ref{fig:Zchannel} shows
the ternary channel which results when duplicating the input ``$1$.''
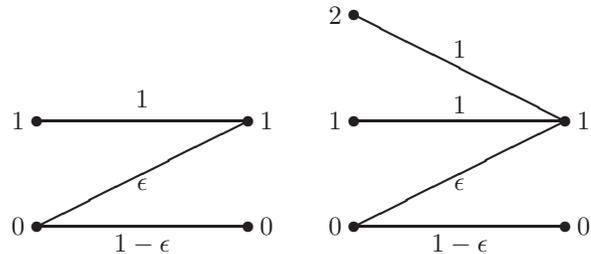
\begin{figure}[!h]
\begin{center}
\begin{picture}(200,90)
\thicklines
\put(0,-40){
\put(0,40){\circle*{4}}
\put(0,40){\line(1,0){80}}
\put(0,40){\line(2,1){80}}
\put(0,80){\circle*{4}}
\put(0,80){\line(1,0){80}}
\put(80,40){\circle*{4}}
\put(80,80){\circle*{4}}
\put(-7,40){\makebox(0,0)[c]{$0$}}
\put(-7,80){\makebox(0,0)[c]{$1$}}
\put(87,40){\makebox(0,0)[c]{$0$}}
\put(87,80){\makebox(0,0)[c]{$1$}}
\put(40,85){{\makebox(0,0)[b]{$1$}}}
\put(40,37){{\makebox(0,0)[t]{$1-\epsilon$}}}
\put(40,58){\makebox(0,0)[t]{$\epsilon$}}
}
\put(120,0){
\put(0,40){\circle*{4}}
\put(0,80){\circle*{4}}
\put(0,40){\line(1,0){80}}
\put(0,0){\circle*{4}}
\put(0,80){\line(2,-1){80}}
\put(0,0){\line(1,0){80}}
\put(0,0){\line(2,1){80}}
\put(80,40){\circle*{4}}
\put(80,0){\circle*{4}}
\put(-7,0){\makebox(0,0)[c]{$0$}}
\put(-7,40){\makebox(0,0)[c]{$1$}}
\put(-7,80){\makebox(0,0)[c]{$2$}}
\put(87,40){\makebox(0,0)[c]{$1$}}
\put(87,0){\makebox(0,0)[c]{$0$}}
\put(40,64){{\makebox(0,0)[b]{$1$}}}
\put(40,43){{\makebox(0,0)[b]{$1$}}}
\put(40,-3){{\makebox(0,0)[t]{$1-\epsilon$}}}
\put(40,18){\makebox(0,0)[t]{$\epsilon$}}
}
\end{picture}
\end{center}
\caption{\label{fig:Zchannel} The Z-channel and its corresponding augmented
channel with ternary input alphabet.}
\end{figure}
Note that the capacity-achieving input distribution for this ternary-input channel is the
uniform one. 
Assume that we can construct a ternary polar code which achieves the
symmetric mutual information of this new channel. (For binary-input channels it
was shown by Ar\i kan \cite{Ari08} that one can achieve the symmetric mutual information and 
there is good reason to believe that an equivalent result holds for $q$-ary input channels.)
Then this gives rise to a capacity-achieving 
coding scheme for the original binary $Z$-channel by
mapping the ternary set $\{0, 1, 2\}$ into the binary set $\{0, 1\}$ in the following way; 
$\{1, 2\} \mapsto 1$ and $0 \mapsto 0$.

More generally, by augmenting the input alphabet and constructing
a code for the extended alphabet, we can achieve rates
arbitrarily close to the capacity of a $q$-ary DMC, assuming
only that we know how to achieve the symmetric mutual information. 

A similar remark applies to the setting of source coding. 
By extending the reconstruction alphabet if necessary and
by using only test channels that induce a uniform distribution
on this extended alphabet one can achieve a rate-distortion performance
arbitrarily close to the Shannon bound, assuming only that 
for the uniform case we can get arbitrarily close.

The previous discussion shows that perhaps the most important generalization
is the construction of polar codes for both source and channel coding for the setting of $q$-ary alphabets.

In Section~\ref{sec:extensions} we have considered some scenarios beyond basic source coding.
E.g., we considered binary versions of the Wyner-Ziv problem as well as the
Gelfand-Pinsker problem. This list is by no means exhaustive.

One possible further generalization is to have source codes with a faster
convergence speed. In \cite{KSU09} it was shown that, by considering larger
matrices (instead of $G_2$), it is possible to obtain better exponents for the
block error probability of the channel coding problem. Such a generalization for source coding
would result in better exponents in the convergence of the average distortion to the
design distortion.

\section*{Acknowledgment}
We would like to thank Eren {\c S}a{\c s}o{\u g}lu and 
Emre Telatar for useful discussions during the development of this paper.
In particular, we would like to thank Emre for his help in proving Lemma~\ref{lem:lowerboundZ}.

\appendix
The proof of \eqref{eqn:Zestimate1} and \eqref{eqn:Zestimate2} is based on the
following approach. For any channel $W:\mathcal{X} \to \mathcal{Y}$ the channels
$W^{[i]} : \mathcal{X} \to \mathcal{Y}\times\mathcal{Y}\times U_0^{i-1}$ are
defined as follows. Let $W^{[0]}$ denote the channel law
\begin{align*}
W^{[0]}(y_0,y_1\mid u_0) = \frac12 \sum_{u_1} W(y_0\mid u_0\oplus u_1)
W(y_1\mid u_1),
\end{align*}
and let $W^{[1]}$ denote the channel law
\begin{align*}
W^{[1]}(y_0,y_1, u_0\mid u_1) =\frac12 W(y_0\mid u_0\oplus u_1)
W(y_1\mid u_1).
\end{align*}

Define a random variable $W_n$ through a tree process
$\{W_n; n\geq 0\}$ with
\begin{align*}
W_0 & = W, \\
W_{n+1} & = W_{n}^{[B_{n+1}]},
\end{align*}
where $\{B_n; n\geq 1\}$ is a sequence of i.i.d.\ random variables defined on a
probability space $(\Omega, \mathcal{F}, \mu )$, and where $B_n$ is a symmetric 
Bernoulli random variable. Defining $\mathcal{F}_0 =
\{\emptyset, \Omega\}$ and $\mathcal{F}_n = \sigma (B_1,\dotsc, B_n)$ for $n
\geq 1$, we augment the above process by the process
$\{Z_n; n\geq 0\} := \{Z(W_n); n\geq 0\}$. The relevance of this process is that 
$W_n \in \{W^{(i)}\}_{i=0}^{2^n-1}$ and moreover the symmetric distribution of
the random variables $B_i$ implies

\begin{align}
\Pr(Z_n\in (a,b)) & =\frac{\lvert\left\{i\in\{0,\dotsc,2^n-1\}: Z^{(i)} \in
(a, b) \right\} \rvert}{2^n}.\label{eqn:Z-count}
\end{align}
In \cite{Ari08} it was shown that 
\begin{align*}
\lim_{n\to \infty }\Pr(Z_n < {2^{-5n/4}}) =   I(W).
\end{align*}
which implies \eqref{eqn:Zestimate1}. 
In \cite{ArT08} the polynomial decay (in terms of $N=2^n$) was improved to exponential decay as
stated below.
\begin{theorem}[Rate of $Z_n$ Approaching $0$ \cite{ArT08}]\label{thm:rateZ}
Given a B-DMC $W$, and any $\beta < \frac{1}{2}$,
\begin{align*}
\lim_{n \to \infty} \Pr(Z_n \leq 2^{-2^{n\beta}}) = I(W).
\end{align*}
\end{theorem}
Of course, this implies \eqref{eqn:Zestimate2}.
For lossy source compression, the important quantity is the  rate at which the random variable 
$Z_n$ approaches $1$ (as compared to $0$).
Let us now show the result mirroring Theorem~\ref{thm:rateZ} for this case, using similar techniques as in \cite{ArT08}.
\begin{theorem}[Rate of $Z_n$ Approaching $1$]\label{thm:rateZbar}
Given a B-DMC $W$, and any $\beta < \frac{1}{2}$,
\begin{align*}
\lim_{n \to \infty} \Pr(Z_n \geq 1- 2^{-2^{n\beta}}) = 1-I(W).
\end{align*}
\end{theorem}
\begin{proof}
Using Lemma~\ref{lem:lowerboundZ} the random variable $Z_{n+1}$ can be bounded as,
\begin{align*}
Z_{n+1} &\geq \sqrt{2Z_n^2 - Z_n^4} \text{ w.p. }\frac12,\\
Z_{n+1} &= Z_n^2 \text{ w.p. }\frac12. 
\end{align*}
Then, with probability $\frac12$,
$Z_{n+1}^2 \geq {1-(1-Z_n^2)^2}$. This implies that
$1-Z_{n+1}^2 \leq (1-Z_{n}^2)^2$.
Similarly, with probability $\frac12$,
\begin{align*}
1-Z_{n+1}^2 &= 1-Z_n^4 \leq 2 (1-Z_n^2).
\end{align*}
Let $X_n$ denote $X_n = 1-Z_n^2$. Then $\{X_n:n\geq 0\}$
satisfies
\begin{align*}
X_{n+1} \leq X_n^2 \text{ w.p. }\frac12,\\
X_{n+1} \leq 2X_n \text{ w.p. }\frac12.
\end{align*}
By adapting the proof of \cite{ArT08}, we can show that for any $\beta <
\frac12$,
\begin{align*}
\lim_{n\to\infty}\Pr(X_n \leq 2^{-2^{n\beta}}) = 1-I(W).
\end{align*}
Using the relation $X_n = 1- Z_n^2  \geq 1- Z_n$, we get 
\begin{align*}
\lim_{n\to\infty}\Pr(1- Z_n \leq 2^{-2^{n\beta}}) = 1-I(W).
\end{align*}
\end{proof}
\begin{lemma}[Lower Bound on $Z$]\label{lem:lowerboundZ}
Let $W_1$ and $W_2$ be two B-DMCs and let $X_1$ and $X_2$ be their inputs with a
uniform prior. Let $Y_1\in \mathcal{Y}_1$ and $Y_2 \in \mathcal{Y}_2$ denote the
outputs.
Let $W$ denote the channel between $X = X_1\oplus X_2$ and the output $(Y_1,Y_2)$, i.e.,
\begin{align*}
W(y_1,y_2\mid x) = \frac12 \sum_{u}W_1(y_1\mid x \oplus
u)W_2(y_2\mid u).
\end{align*}
Then 
\begin{align*}
Z(W) \geq \sqrt{Z(W_1)^2 + Z(W_2)^2 - Z(W_1)^2Z(W_2)^2}.
\end{align*}
\end{lemma}
\begin{proof}
Let $Z = Z(W)$ and $Z_i = Z(W_i)$. $Z$ can be expanded as follows.
\begin{align*}
&Z =  \sum_{y_1,y_2}\sqrt{W(y_1,y_2\mid 0) W(y_1,y_2\mid 1) }\\
& = \frac{1}{2}\sum_{y_1,y_2} \Big[W_1(y_1 \mid 0)W_2(y_2\mid 0) W_1(y_1 \mid 0)W_2(y_2\mid 1) \\
&\phantom{===}+ W_1(y_1 \mid 0)W_2(y_2\mid 0)W_1(y_1 \mid 1)W_2(y_2\mid 0)\\ 
&\phantom{===}+ W_1(y_1 \mid 1)W_2(y_2\mid 1)W_1(y_1 \mid 0)W_2(y_2\mid 1)\\ 
&\phantom{===}+ W_1(y_1 \mid 1)W_2(y_2\mid 1)W_1(y_1 \mid 1)W_2(y_2\mid 0)\Big]^{\frac12}\\
& = \frac{Z_1Z_2}{2} \sum_{y_1,y_2} P_1(y_1) P_2(y_2)\\
&\sqrt{\frac{W_1(y_1\mid 0)}{W_1(y_1\mid 1)}+
\frac{W_1(y_1\mid 1)}{W_1(y_1\mid 0)}+ \frac{W_2(y_2\mid 0)}{W_2(y_2\mid 1)}+
\frac{W_2(y_2\mid 1)}{W_2(y_2\mid 0)}}
\end{align*}
where $P_i(y_i)$ denotes
\begin{align*}
P_i(y_i) = \frac{\sqrt{W_i(y_i\mid 0)W_i(y_i\mid 1)}}{Z_i}.
\end{align*}
Note that $P_i$ is a probability distribution over $\mathcal{Y}_i$.
Let $\bE_i$ denote the expectation with respect to $P_i$ 
and let
\begin{align*}
A_i(y) \triangleq \sqrt{\frac{W_i(y\mid 0)}{W_i(y\mid 1)}} +
\sqrt{\frac{W_i(y\mid 1)}{W_i(y\mid 0)}}.
\end{align*}
Then $Z$ can be expressed as 
\begin{align*}
Z = \frac{Z_1Z_2}{2}\bE_{1,2}\left[\sqrt{\left(A_1(Y_1)\right)^2
+\left(A_2(Y_2)\right)^2-4 }\right].
\end{align*}
The arithmetic-mean geometric-mean inequality implies that $A_i(y) \geq 2$.
Therefore, for any $y_i \in \mathcal{Y}_i$, $A_i(y_i)^2 -4 \geq 0$.
Note that the function $f(x) = \sqrt{x^2+a}$ is convex for $a\geq 0$. Applying Jensen's inequality
first with respect to the expectation $\bE_1$ and then with respect to $\bE_2$, we get
\begin{align*}
Z &\geq \frac{Z_1Z_2}{2}\bE_2\left[\sqrt{\left(\bE_1\left[A_1(Y_1)\right]\right)^2
+\left(A_2(Y_2)\right)^2 -4}\right]\\
&\geq
\frac{Z_1Z_2}{2}\sqrt{\left(\bE_1\left[A_1(Y_1)\right]\right)^2+\left(\bE_2\left[A_2(Y_2)\right]\right)^2-4}.
\end{align*}
The claim follows by substituting $\bE_i[A_i(Y_i)] = \frac{2}{Z_i}$.
\end{proof}

\bibliographystyle{IEEEtran} 
\bibliography{lth,lthpub}
\end{document}

%% file: definition.tex
\definecolor{TODO}{rgb}{0.6,0.6,0.6} 

\definecolor{TOCHECK}{rgb}{0.8,0.8,0.8} 

\newtheorem{theorem}{Theorem}

\newcommand{\btheo}{\begin{theorem}}
\newcommand{\etheo}{\end{theorem}}
\newcommand{\bproof}{\begin{proof}}
\newcommand{\eproof}{\end{proof}}
\newtheorem{definition}[theorem]{Definition}
\newcommand{\bdefi}{\begin{definition}}
\newcommand{\edefi}{\end{definition}}
\newtheorem{fact}[theorem]{Fact}
\newcommand{\bprop}{\begin{fact}}
\newcommand{\eprop}{\end{fact}}
\newtheorem{corollary}[theorem]{Corollary}
\newcommand{\bcor}{\begin{corollary}}
\newcommand{\ecor}{\end{corollary}}
\newtheorem{example}[theorem]{Example}
\newcommand{\bex}{\begin{example}}
\newcommand{\eex}{\end{example}}
\newtheorem{lemma}[theorem]{Lemma}
\newcommand{\blemma}{\begin{lemma}}
\newcommand{\elemma}{\end{lemma}}
\newtheorem{remark}[theorem]{Remark}
\newcommand{\bremark}{\begin{remark}}
\newcommand{\eremark}{\end{remark}}
\newtheorem{conj}[theorem]{Conjecture}
\newcommand{\bconj}{\begin{conj}}
\newcommand{\econj}{\end{conj}}

\def\0{{\tt 0}} 
\def\1{{\tt 1}} 
\def\?{{\tt *}} 
\newcommand{\qed}{{\hfill \footnotesize $\blacksquare$}}
\renewcommand{\mid}{\,|\,}

\newcommand {\bE} {\mathbb{E}}

\newcommand{\ind}{\mathbbm{1}}
\newcommand{\indicator}[1]{\ind_{\{ #1 \}}}

\newcommand{\bY}{\bar{Y}}
\newcommand{\bX}{\bar{X}}
\newcommand{\bZ}{\bar{Z}}
\newcommand{\by}{\bar{y}}
\newcommand{\bx}{\bar{x}}

\newcommand{\hu}{\hat{u}}
\newcommand{\bu}{\bar{u}}
\newcommand{\be}{\bar{e}}

\newcommand{\bv}{\bar{v}}
\newcommand{\bV}{\bar{V}}
\newcommand{\bU}{\bar{U}}

\newcommand {\disto}{\mathtt{d}}

\newcommand {\WZ}{\text{\tiny WZ}}
\newcommand {\GP}{\text{\tiny GP}}
\newcommand {\code}{\mathtt{C}}

\newcommand {\Pb}{P^B}
\newcommand {\quant}{\mathcal{Q}}

%% file: ps/transform.tex
\setlength{\unitlength}{0.5bp}%
\begin{picture}(320,200)(0,0)
\put(0,0){\includegraphics[scale=0.5]{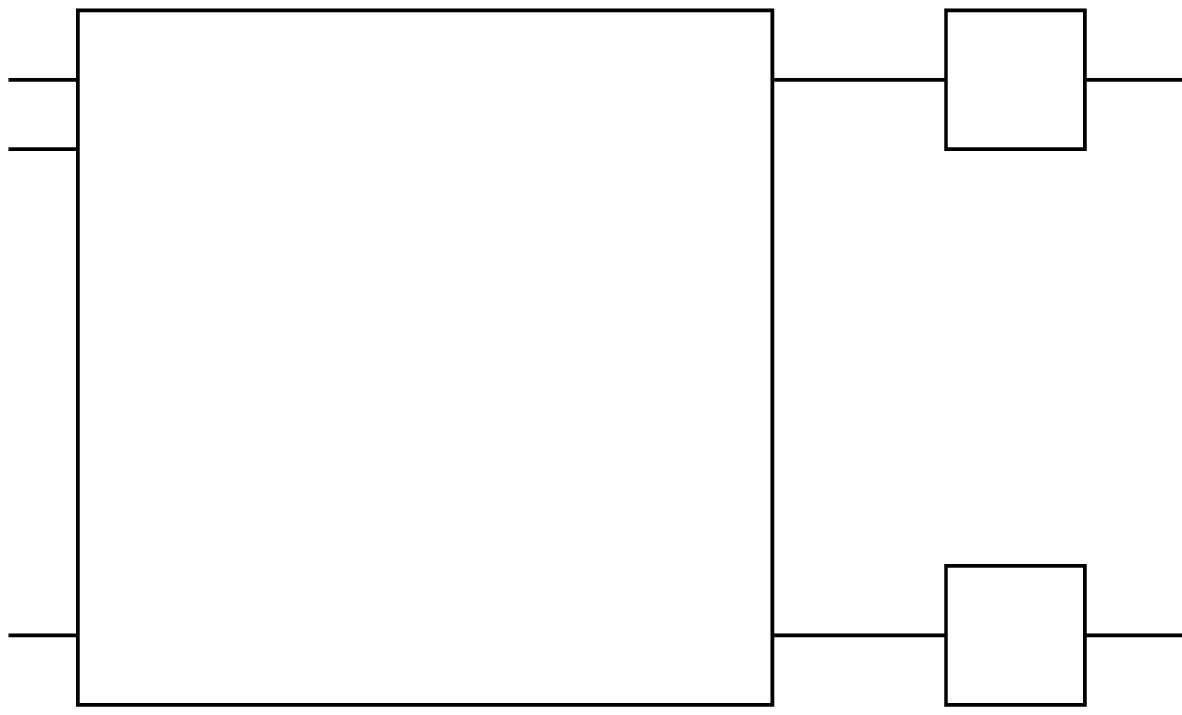}}
\put(282,185){\footnotesize $W$}
\put(285,145){\footnotesize$\cdot$}
\put(285,105){\footnotesize$\cdot$}
\put(285,65){\footnotesize$\cdot$}
\put(282,25){\footnotesize$W$}
\put(105,110){\footnotesize$A_nG^{\otimes n}$}
\put(-40,185){\footnotesize ${U}_0$}
\put(-40,165){\footnotesize ${U}_1$}
\put(-30,130){\footnotesize $\cdot$}
\put(-30,95){\footnotesize $\cdot$}
\put(-30,60){\footnotesize $\cdot$}
\put(-40,25){\footnotesize $U_{N-1}$}
\put(225,195){\tiny ${X}_0$}
\put(225,35){\tiny $X_{N-1}$}
\put(340,185){\footnotesize ${Y}_0$}
\put(340,130){\footnotesize $\cdot$}
\put(340,95){\footnotesize $\cdot$}
\put(340,60){\footnotesize $\cdot$}
\put(340,25){\footnotesize $Y_{N-1}$}
\end{picture}

%% file: ps/trellis_rev.tex
\setlength{\unitlength}{0.65bp}%
\small
\begin{picture}(250,230)(10,5)
\put(5,0){\includegraphics[scale=0.65]{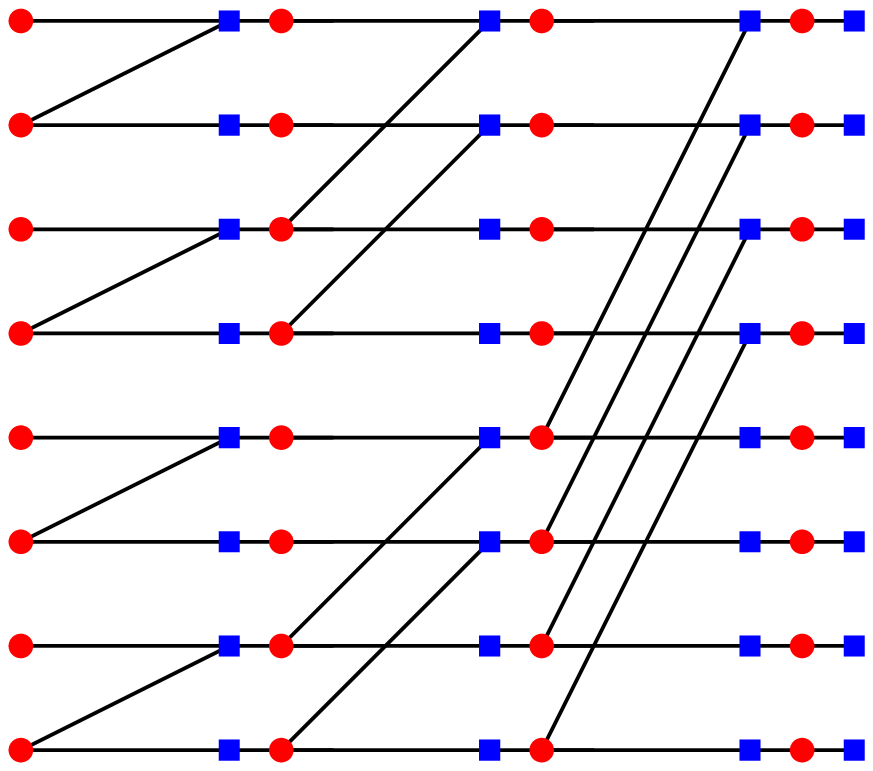}}
\put(210,10){
\put(15,3){$X_7$}
\put(15,33){$X_6$}
\put(15,63){$X_5$}
\put(15,93){$X_4$}
\put(15,123){$X_3$}
\put(15,153){$X_2$}
\put(15,183){$X_1$}
\put(15,213){$X_0$}
}
\put(-15,0){$U_7$}
\put(-15,30){$U_3$}
\put(-15,60){$U_5$}
\put(-15,90){$U_1$}
\put(-15,120){$U_6$}
\put(-15,150){$U_2$}
\put(-15,180){$U_4$}
\put(-15,210){$U_0$}

\put(-5,-3){
\put(260,3){$W(y_7\mid x_7)$}
\put(290,53){$\vdots$}
\put(290,113){$\vdots$}
\put(290,173){$\vdots$}
\put(260,213){$W(y_0\mid x_0)$}
}

\end{picture}

%% file: ps/LOSSY.tex
\setlength{\unitlength}{1.75bp}%
\begin{picture}(130,85)(-5,-5)
{
\put(0,0){\includegraphics[scale=1.75]{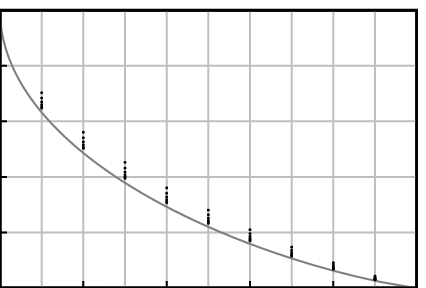}}
\multiputlist(0,-8)(24,0)[b]{$0.0$,$0.2$,$0.4$, $0.6$,$0.8$}
\multiputlist(-12,16)(0,16)[l]{$0.1$,$0.2$,$0.3$,$0.4$}
\put(-12,80){\makebox(0,0)[tl]{$D$}}
\put(120,-8){\makebox(0,0)[br]{$R$}}
}
\end{picture}

%% file: ps/wynerziv.tex
\setlength{\unitlength}{1bp}%
\begin{picture}(220,70)(0,0)
\put(0,0){\includegraphics[scale=1]{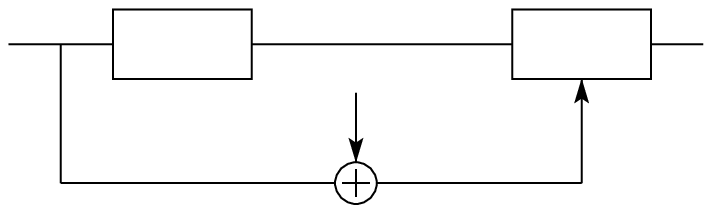}}
\put(0,0)
{
\put(105,32){\makebox(0,0)[c]{$Z$}}
\put(60,50){\makebox(0,0)[c]{Encoder}}
\put(175,50){\makebox(0,0)[c]{Decoder}}
\put(115,60){\makebox(0,0)[c]{$R$}}
\put(25,60){\makebox(0,0)[c]{$Y$}}
\put(90,60){\makebox(0,0)[c]{$X$}}
\put(185,30){\makebox(0,0)[c]{$Y'$}}

}
\end{picture}

%% file: ps/gelfandpinsker.tex
\setlength{\unitlength}{1bp}%
\begin{picture}(220,60)(0,0)
\put(0,0){\includegraphics[scale=1]{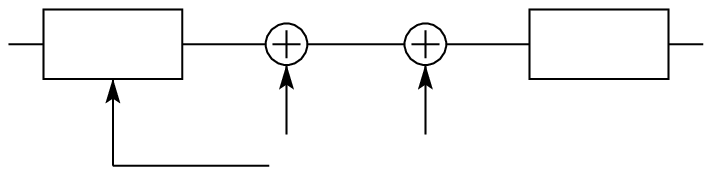}}
\put(0,0)
{
\put(130,18){\makebox(0,0)[c]{$Z$}}
\put(90,18){\makebox(0,0)[c]{$S$}}
\put(40,50){\makebox(0,0)[c]{Encoder}}
\put(180,50){\makebox(0,0)[c]{Decoder}}
\put(70,60){\makebox(0,0)[c]{$X$}}
\put(10,60){\makebox(0,0)[c]{$M$}}
\put(210,60){\makebox(0,0)[c]{$\hat{M}$}}
\put(150,60){\makebox(0,0)[c]{$Y$}}

}
\end{picture}

%% file: ps/onehelper.tex
\setlength{\unitlength}{1bp}%
\begin{picture}(220,70)(0,0)
\put(0,-10){\includegraphics[scale=1]{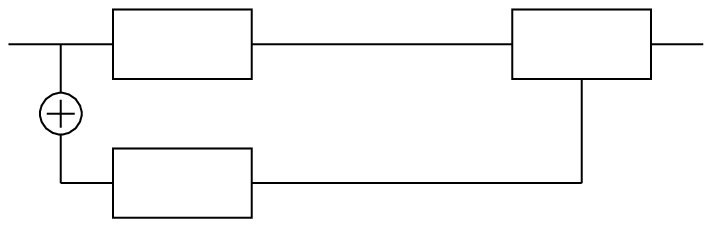}}
\put(0,0)
{
\put(15,32){\makebox(0,0)[c]{$Z$}}
\put(60,50){\makebox(0,0)[c]{Encoder}}
\put(60,10){\makebox(0,0)[c]{Helper}}
\put(175,50){\makebox(0,0)[c]{Decoder}}
\put(115,60){\makebox(0,0)[c]{$R$}}
\put(115,20){\makebox(0,0)[c]{$R'$}}
\put(25,60){\makebox(0,0)[c]{$Y$}}
\put(25,0){\makebox(0,0)[c]{$Y'$}}
\put(90,60){\makebox(0,0)[c]{$X$}}
\put(90,20){\makebox(0,0)[c]{$X'$}}
\put(205,60){\makebox(0,0)[c]{$\hat{Y}$}}

}
\end{picture}